\documentclass[11pt,a4paper]{amsart}
\usepackage{amssymb,amsmath}
\usepackage[alphabetic]{amsrefs}
\usepackage{amsmath,amssymb}
\usepackage{mathrsfs}
\usepackage[dvipdfmx]{graphicx,color}
\usepackage{wrapfig}
\usepackage{calc}
\usepackage{extarrows}
\usepackage{amsthm}
\usepackage{latexsym}
\usepackage{slashed}
\usepackage{framed}
\usepackage{ascmac}
\usepackage[all]{xy}
\usepackage{hyperref}
\usepackage[alphabetic]{amsrefs}
\usepackage[mathscr]{eucal}

\makeatletter
\@addtoreset{equation}{section}
\makeatother

\DeclareMathAlphabet{\mathfrak}{U}{euf}{m}{n}
\SetMathAlphabet{\mathfrak}{bold}{U}{euf}{b}{n}

\theoremstyle{definition}
\newtheorem{defn}[equation]{Definition}
\newtheorem{assump}[equation]{Assumption}
\theoremstyle{plain}
\newtheorem{thm}[equation]{Theorem}
\newtheorem{prp}[equation]{Proposition}
\newtheorem{lem}[equation]{Lemma}
\newtheorem{cor}[equation]{Corollary}

\theoremstyle{remark}
\newtheorem{remk}[equation]{Remark}
\newtheorem{exmp}[equation]{Example}
\newtheorem*{remk*}{Remark}
\newtheorem{blank}[equation]{}

\newcommand{\ma}[1]{\begin{align*} #1 \end{align*}}
\newcommand{\maa}[1]{\begin{align} #1 \end{align}}

\newcommand{\real}{\mathbb R}
\newcommand{\comp}{\mathbb C}
\newcommand{\zahl}{\mathbb Z}
\newcommand{\quot}{\mathbb Q}
\newcommand{\nat}{\mathbb N}

\DeclareMathOperator*{\Hom}{Hom}

\DeclareMathOperator*{\Aut}{Aut}

\newcommand{\id}{\mathrm{id}}
\DeclareMathOperator*{\rank}{rank}
\DeclareMathOperator*{\Ker}{Ker}

\DeclareMathOperator*{\supp}{supp}

\DeclareMathOperator{\hotimes}{\hat{\otimes}}
\DeclareMathOperator{\Irr}{Irr}

\newcommand{\ebk}[1]{\left<  #1 \right>}

\newcommand{\ssbk}[1]{\left\| #1 \right\|}

\newcommand{\pmx}[1]{\begin{pmatrix} #1 \end{pmatrix}}%

\newcommand{\Cliff}{\mathbb{C}\ell }
\newcommand{\pt}{{\rm pt}}

\DeclareMathOperator{\Tr}{Tr}

\newcommand{\xra}{\xrightarrow}

\DefineSimpleKey{bib}{howpublished}

\renewcommand{\PrintDOI}[1]{%
  \href{http://dx.doi.org/#1}{{\tt DOI:#1}}%
  \IfEmptyBibField{volume}{, (to appear in print)}{}%
}

\renewcommand{\eprint}[1]{#1}

\BibSpec{book}{%
    +{}  {\PrintPrimary}                {transition}
    +{.} { \textit}                     {title}
    +{.} { }                            {part}
    +{:} { \textit}                     {subtitle}
    +{,} { \PrintEdition}               {edition}
    +{}  { \PrintEditorsB}              {editor}
    +{,} { \PrintTranslatorsC}          {translator}
    +{,} { \PrintContributions}         {contribution}
    +{,} { }                            {series}
    +{,} { \voltext}                    {volume}
    +{,} { }                            {publisher}
    +{,} { }                            {organization}
    +{,} { }                            {address}
    +{,} { }                            {status}
    +{,} { \PrintDOI}                   {doi}
    +{,} { \PrintISBNs}                 {isbn}
    +{;} { \PrintReprint}               {reprint}
    +{,} { \PrintDate}                  {date}
    +{.} { }                            {note}
    +{.} {}                             {transition}
}

\BibSpec{article}{%
    +{}  {\PrintAuthors}                {author}
    +{,} { \textit}                     {title}
    +{.} { }                            {part}
    +{:} { \textit}                     {subtitle}
    +{,} { \PrintContributions}         {contribution}
    +{.} { \PrintPartials}              {partial}
    +{,} { }                            {journal}
    +{}  { \textbf}                     {volume}
    +{}  { \PrintDatePV}                {year}
    +{,} { \issuetext}                  {number}
    +{,} { \eprintpages}                {pages}
    +{,} { }                            {status}
    +{,} { \PrintDOI}                   {doi}
    +{;} { \PrintReprint}               {reprint}
    +{.} { }                            {note}
    +{,} { \eprint}                     {eprint}
    +{.} {}                             {transition}
}
\BibSpec{collection.article}{%
    +{}  {\PrintAuthors}                {author}
    +{,} { \textit}                     {title}
    +{.} { }                            {part}
    +{:} { \textit}                     {subtitle}
    +{,} { \PrintContributions}         {contribution}
    +{,} { \PrintConference}            {conference}
    +{}  {\PrintBook}                   {book}
    +{,} { }                            {booktitle}
    +{,} { \PrintDateB}                 {year}
    +{,} { pp.~}                        {pages}
    +{,} { }                            {publisher}
    +{,} { }                            {organization}
    +{,} { }                            {address}
    +{,} { }                            {status}
    +{,} { \PrintDOI}                   {doi}
    +{,} { \eprint}        {eprint}
    +{}  { \parenthesize}               {language}
    +{}  { \PrintTranslation}           {translation}
    +{;} { \PrintReprint}               {reprint}
    +{.} { }                            {note}
    +{.} {}                             {transition}
}

\BibSpec{misc}{
  +{}{\PrintAuthors}  {author}
  +{,}{ \textit}      {title}
  +{,}{ }             {note}
  +{,}{ }             {date}
  +{.}{}              {transition}
}

\usepackage{accents}
\usepackage{makecell}
\usepackage{enumitem}
\makeatletter
\def\namedlabel#1#2{\begingroup
    #2%
    \def\@currentlabel{#2}%
    \phantomsection\label{#1}\endgroup
}
\let\latexchi\chi
\makeatletter
\renewcommand\chi{\@ifnextchar_\sub@chi\latexchi}
\newcommand{\sub@chi}[2]{
  \@ifnextchar^{\subsup@chi{#2}}{\latexchi^{}_{#2}}%
}
\newcommand{\subsup@chi}[3]{
  \latexchi_{#1}^{#3}%
}
\makeatother
\usepackage{hhline}
\usepackage[mathscr]{eucal}

\newcommand{\Euc}{\mathbb{E}}

\newcommand{\lu}[1]{{}^{#1} \hspace{-0.1ex} }
\renewcommand{\blank}{\text{\textvisiblespace}}

\newcommand{\F}{\mathscr{F}}

\newcommand{\K}{\mathrm{K}}

\newcommand{\KR}{\mathrm{KR}}

\newcommand{\KQ}{\mathrm{KQ}}
\newcommand{\KK}{\mathrm{KK}}

\newcommand{\KKR}{\mathrm{KKR}}

\newcommand{\G}{\mathcal{G}}

\newcommand{\Bop}{\mathbb{B}}

\newcommand{\Kop}{\mathbb{K}}
\newcommand{\Mop}{\mathbb{M}}
\newcommand{\Cst}{\mathrm{C}^*}
\newcommand{\Pen}{\mathrm{Pen}}
\newcommand{\Hilb}{\mathscr{H}}

\newcommand{\vD}{\mathrm{vD}}

\newcommand{\Zt}{\mathbb{Z}_2}
\newcommand{\Cl}{\mathrm{C}\ell}

\newcommand{\cone}{\mathrm{cone}}
\let\Im\relax
\DeclareMathOperator{\Im}{\mathrm{Im}}

\DeclareMathOperator{\Ind}{Ind}
\DeclareMathOperator{\Res}{Res}
\newcommand{\ev}{\mathrm{ev}}
\DeclareMathOperator{\bind}{ind_{\mathit{b}}}
\DeclareMathOperator{\eind}{ind_{\mathit{e}}}

\newcommand{\TP}{\mathscr{TP}}

\newcommand{\lInd}{\ell ^2\mathchar`-\mathrm{Ind}}

\newcommand{\KL}{\mathrm{KL}}

\newcommand{\KAII}{\mathrm{KA\hspace{-.08em}I\hspace{-.08em}I}}

\newcommand{\KCII}{\mathrm{KC\hspace{-.08em}I\hspace{-.08em}I}}
\newcommand{\KBDI}{\mathrm{KB\hspace{-.08em}D\hspace{-.08em}I}}

\newcommand{\KDIII}{\mathrm{KD\hspace{-.08em}I\hspace{-.08em}I\hspace{-.08em}I}}
\newcommand{\typeA}{{\rm A}}
\newcommand{\typeAI}{{\rm A\hspace{-.08em}I}}
\newcommand{\typeAII}{{\rm A\hspace{-.08em}I\hspace{-.08em}I}}
\newcommand{\typeAIII}{{\rm A\hspace{-.08em}I\hspace{-.08em}I\hspace{-.08em}I}}
\newcommand{\typeC}{{\rm C}}
\newcommand{\typeCI}{{\rm C\hspace{-.08em}I}}
\newcommand{\typeCII}{{\rm C\hspace{-.08em}I\hspace{-.08em}I}}
\newcommand{\typeBDI}{{\rm B\hspace{-.08em}D\hspace{-.08em}I}}
\newcommand{\typeD}{{\rm D}}
\newcommand{\typeDIII}{{\rm D\hspace{-.08em}I\hspace{-.08em}I\hspace{-.08em}I}}

\allowdisplaybreaks[3]
\title[Controlled topological phases and bulk-edge correspondence]{Controlled topological phases and bulk-edge correspondence}
\author[Y. Kubota]{Yosuke Kubota}
\address{Graduate School of Mathematical Science, The University of Tokyo, 3-8-1 Komaba, Meguro-ku, Tokyo 153-8914, Japan}
\email{ykubota@ms.u-tokyo.ac.jp}

\date{9 February, 2016}

\begin{document}
\maketitle
\begin{abstract}
In this paper, we introduce a variation of the notion of topological phase reflecting metric structure of the position space. 
This framework contains not only periodic and non-periodic systems with symmetries in Kitaev's periodic table but also topological crystalline insulators. 
We also define the bulk and edge indices as invariants taking values in the twisted equivariant $\K$-groups of Roe algebras as generalizations of existing invariants such as the Hall conductance or the Kane--Mele $\mathbb{Z}_2$-invariant. 
As a consequence, we obtain a new mathematical proof of the bulk-edge correspondence by using the coarse Mayer-Vietoris exact sequence.\end{abstract}

\tableofcontents

\section{Introduction}
In condensed matter physics, it is known that some physical quantities are related to the topology of the Hamiltonian operator, which gives a reason for them to be quantized and robust. 
A classical example is the $2$-dimensional integer quantum Hall effect (IQHE). 
Here, the Hamiltonian operator with translation symmetry is regarded as a continuous family of self-adjoint operators on the Brillouin torus (the Pontrjagin dual of the translation group $\zahl ^2$) by the Fourier transform. 
If the Hamiltonian has a spectral gap at the Fermi level, the family of spectral subspaces corresponding to the gap forms a vector bundle on the Brillouin torus.
The first Chern numbers of this vector bundle, which classifies vector bundles on the $2$-dimensional torus, is called their bulk index and related to the Hall conductance by the TKNN formula \cite{TKNN}. 
On the other hand, from the viewpoint of the edge channel picture, the Hamiltonian with a spectral gap at the Fermi level is insulated in the bulk and an electric flow occurs at the edge of the material. 
More precisely, the edge Hamiltonian obtained by cutting the bulk Hamiltonian by the projection onto the half-plane gives a continuous family of Toeplitz operators on the $1$-dimensional torus.
Its spectral flow is called the edge index and regarded as the number of edge channels and hence the Hall conductance. 
The \textit{bulk-edge correspondence} ensures that these two Hall conductances obtained from different presentations actually coincide. 
This is first proved by Hatsugai using the topology of a Riemann surface \cite{H1993PRB}.


These arguments are understood systematically in terms of $\K$-theory and cyclic cohomology for $\Cst$-algebras, which are basic tools in noncommutative geometry. 
We remark that the Brillouin torus is the same thing as (the Gelfand-Naimark dual of) the group $\Cst$-algebra $C^*_r (\zahl ^2)$. 
Actually, the bulk and edge Hamiltonian operators determine elements in $\K_0(C^*_r(\zahl ^2))$ and $\K_1(C^*_r(\zahl))$ respectively. 
Moreover, the TKNN formula is translated into the language of noncommutative geometry as the pairing between these elements and the fundamental cyclic cocycle by Bellissard~\cite{MR965981}.
In \cite{MR1877916}, a simple proof of the bulk-edge correspondence is given by using the boundary map of the six-term exact sequence of $\K$-groups associated to the Toeplitz exact sequence.


A merit of this noncommutative geometric approach is that it can be applied for more general settings. 
Particularly, it is remarkable that it enables us to deal with the case that the translation symmetry is broken by disorder \citelist{\cite{MR862832}\cite{MR1303779}\cite{MR1295473}\cite{MR1877916}\cite{MR2084009}}. 
Here, the study of the topology of quantum systems is divided into two problems: classification of topological phases (homotopy classes of gapped quantum systems) and calculation of the indices, which are related to physical quantities.
Actually, in this case topological phases are no longer classified by the index but by the $\K$-group of the crossed product $\zahl ^d \ltimes C(\Omega)$, which plays the role of the noncommutative Brillouin torus.
Here, $\Omega$ is a certain compactification of $\zahl ^d$ which is regarded as the probabilistic space of random Hamiltonians. 
The bulk index is defined in the same way as the translation-invariant case by using a cyclic cocycle. 
In \cite{MR1877916}, the corresponding edge index is also defined and the bulk-edge correspondence is proved by using the Toeplitz extension for crossed products by $\zahl$. 
As for bulk systems, they are generalized by Bellissard and his collaborators \citelist{\cite{MR1221111}\cite{MR1798994}\cite{MR2009996}} to materials without any translation symmetry on its distribution of atoms such as quasi-crystals. They adopt a certain groupoid $\Cst$-algebra as the noncommutative Brillouin torus.


In the recent study of topological insulators and topological superconductors, some of their topological invariants take values in $\Zt$ instead of $\zahl$. 
For example, consider $2$-dimensional spin quantum Hall systems~\cite{PhysRevLett.95.226801}, that is, a quantum system with time-reversal symmetry (an antilinear unitary operator) whose square is $-1$. 
Here, the vector bundle given by a spectral subspace of the bulk Hamiltonian determines an element in the $\KQ ^0$-group of the Brillouin torus with the canonical involution, which is isomorphic to $\zahl \oplus \Zt$. 
The corresponding element in the second component $\Zt$ is called the \emph{Kane--Mele invariant}. 
In this case, the bulk-edge correspondence ensures that this invariant coincides with the parity of helical edge channels. 
It is formulated and proved by Avila--Schulz-Baldes--Villegas-Blas~\cite{MR3063955} and Graf--Porta~\cite{MR3123539} by means of an additional symmetry coming from spin operators and the topology of Riemann surfaces like \cite{H1993PRB} respectively. 
Moreover, a $\Zt$-invariant is also defined for disordered systems in \cite{GrossmannSchulz-Baldes} and \cite{mathph150805485}.


More generally, there are three fundamental operators in quantum mechanics: particle-hole symmetry $C$ (odd and antilinear), time-reversal symmetry $T$ (even and antilinear) and chiral (sub-lattice) symmetry $P$ (odd and linear). 
By Altland--Zirnbauer~\cite{AZ1997} and Kitaev~\cite{K2009}, symmetries of quantum mechanics generated by some of these operators are classified in $10$ types, each of which corresponds to one of $2$ complex and $8$ real Clifford algebras. 
For each of them, classifications of topological phases (for both periodic and nonperiodic systems) are given by Thiang~\cite{mathph14067366} and Kellendonk~\cite{mathKT150906271} by using Karoubi's $\K$-theory and van Daele's $\K$-theory respectively.
The notion of (bulk) index for them is introduced by Gro{\ss}mann--Schulz-Baldes~\cite{GrossmannSchulz-Baldes} as the pairing with a $\K$-cycle. 
For the bulk-edge correspondence, there is a $\KK$-theoretic approach by Bourne--Carey--Rennie~\cite{mathKT150907210}.
In addition to them, a more detailed classification for periodic systems is given by using homotopy theory in a series of works by Nittis--Gomi~\citelist{\cite{MR3282332}\cite{MR3366050}\cite{mathph150404863}}.


In \cite{MR3119923}, Freed and Moore studied quantum systems with symmetries given by finite extensions of translation groups such as topological crystalline insulators \citelist{\cite{PhysRevLett.106.106802}\cite{PhysRevB.88.075142}}. 
A remarkable feature of their approach is that they developed the theory in an axiomatic way starting from a small number of general principles.
According to Wigner's theorem~\cite{MR0106711} (see also Section 1 of \cite{MR3119923}), a symmetry of quantum mechanics is given by a twisted (i.e.\ even/odd projective unitary/antiunitary) representation of a group on a $\Zt$-graded complex Hilbert space, which are classified by a triplet $(\phi,c,\tau)$ of group cocycles called twists (Definition \ref{defn:tw}).
From this viewpoint, a CT-symmetry (a symmetry of quantum mechanics generated by some of CPT operators) is characterized as a certain twist of a subgroup $\mathscr{A}$ of $\Zt \times \Zt$. 
Actually, the Altland--Zirnbauer classification is reproved in Proposition 6.4 of \cite{MR3119923} in terms of the group cohomology.
They gave a classification of topological phases with the symmetry given by $(G,\phi,c,\tau)$ in terms of twisted equivariant $\K$-theory of groupoids when $G$ is a Galilean symmetry.


In this paper, we introduce a new framework of topological phases, which is a relevant generalization of the one introduced in \cite{MR3119923} for systems with an arbitrary symmetry of quantum mechanics such as translations (not necessarily cocompact), rotations and their combinations.

First, for simplicity we describe the case that the system does not have any symmetry.
We focus on the property of tight-binding Hamiltonians that their propagations are controlled at the position space (or distribution of atoms) represented by a discrete metric space $X$. 
Actually, in previous researches such as \cite{MR1295473}, Hamiltonian operators are assumed to be approximated in the norm topology by operators whose propagations (the maximum of $d(x,y)$ such that $H_{x,y}\neq 0$) are finite. 
This assumption corresponds to continuity of its Fourier transform in the case that $H$ is translation invariant. 
We say that such a Hamiltonian is a bulk Hamiltonian \emph{controlled at $X$} and two Hamiltonians are in the same \emph{controlled topological phase} if they are homotopic (more precisely, see Definition \ref{defn:bulkTP}).

The norm closure $C^*_u(X)$ of the algebra of bounded operators with finite propagation on $\ell ^2(X)$ is called the uniform Roe algebra of $X$. 
This $\Cst$-algebra is similar to the noncommutative Brillouin torus. 
Actually, when $X=\zahl ^d$ as a metric space, $C^*_u(X)$ is isomorphic to the crossed product $\zahl ^d \ltimes C(\beta X)$ where $\beta X$ is the Stone--\v{C}ech compactification of $X$ (Example \ref{exmp:groupoid}). 
Now, the set of controlled topological phases is nothing but $\K_0(C^*_u(X))$ as in the noncommutative geometric approach for disordered systems (Proposition \ref{prp:bulk}).

The uniform Roe algebra of a metric space is one of central subjects in coarse geometry, a part of geometry studying ``large-scale'' properties of metric spaces and identifying two spaces that look the same from a large distance such as the Euclidean space and its lattice. 
Actually, the set of topological phases is flexible with respect to deformation of distribution of atoms. 
More precisely, it is independent of the choice of discrete subsets of the Euclidean space which is distributed uniformly (i.e.\ Delone sets, Definition \ref{defn:Delone}).

Unfortunately, the $\K$-group of the uniform Roe algebra is in general highly complicated and often has uncountably many basis. 
On the other hand, $\K$-theory of another $\Cst$-algebra called the Roe algebra, which is originally introduced by Roe~\citelist{\cite{MR918459}\cite{MR1147350}} for the purpose of studying index theory of non-compact manifolds, is easy to handle and widely studied.
A useful tool is the (equivariant) coarse Baum--Connes map~\citelist{\cite{MR1388312}\cite{MR1451759}\cite{MR2565716}}, which gives an isomorphism between a certain $\K$-homology groups and $\K$-groups of Roe algebras (note that there is another use of the Baum-Connes isomorphism in connection with the T-duality studied in \citelist{\cite{1751-8121-48-42-42FT02}\cite{hep-th150505250}\cite{hep-th150604492}\cite{hep-th151004785}}).
For example, when $X$ is a lattice of $\Euc ^d$, the groups $\K _*(C^*(X))$ are isomorphic to $\K^{d-*}(\pt )$.
Here, we consider the canonical homomorphism of $\K$-groups induced from the inclusion $C^*_u(X) \subset C^*(X)$.
Roughly speaking, it gives a rough classification of topological phases, which captures the primary part of the topology of quantum systems (see Remark \ref{remk:conc}). 
Actually, it coincides with the Hall conductance defined by the pairing with the fundamental cyclic cocycle.
Therefore, we call it the \emph{bulk index}.


Another powerful tool for the study of the $\K$-theory of (uniform) Roe algebras is the coarse Mayer--Vietoris exact sequence \cite{MR1219916}. 
Fortunately, the boundary map of the coarse Mayer--Vietoris exact sequence for a partition $X=Y_+ \cup Y_-$ provides a good explanation of the bulk-edge correspondence: it maps a bulk Hamiltonian $H$ to the corresponding edge Hamiltonian $\hat{H}$ obtained by cutting $H$ by the support projection onto $Y_1$.
Following this, we give a definition of edge quantum systems controlled at a pair $Z \subset Y$ of metric spaces (corresponding to the edge and the bulk of a material) in Definition \ref{defn:edgeTP}. 
The set of controlled edge topological phases is actually isomorphic to (a subgroup of) the $\K_{-1}$-group of the relative uniform Roe algebra $C^*_u(Z \subset Y)$ (Definition \ref{defn:Roe}). 
The edge index is also defined in the same way as the bulk case. 
Finally, the bulk-edge correspondence is a consequence of the fact that the Mayer--Vietoris boundary map is an isomorphism for good partitions.
This formulation is more intuitive than the previous approach using Toeplitz extensions for $\zahl$-crossed products because the edge algebra $\zahl ^{d-1} \ltimes C(\Omega)$ is a subalgebra of infinite direct product of $C^*_u(Z)$ (see Corollary \ref{cor:Toep} and its proof).

The same arguments work for the study of quantum systems with symmetries of quantum mechanics.
We deal with the situation that a discrete group $G$ acts in the position space properly and isometrically. 
Under an assumption on twists (Assumption \ref{assump:ext}), we show that the bulk and edge topological phases are classified by a modification of the twisted equivariant $\K$-theory of invariant uniform Roe algebras (Definition \ref{defn:Roe}).
Here, we use the twisted equivariant $\K$-group for $\Cst$-algebras defined in \cite{Kubota1} as a generalization of van Daele's $\K$-theory.
The bulk and edge indices take values in a modification of the twisted equivariant $\K$-group of the invariant Roe algebras, which are calculated by the equivariant coarse Baum-Connes isomorphism (Theorem \ref{thm:cBC}).

As long as considering about CT-symmetries, the twisted equivariant $\K$-group is the same thing as $DK$-groups in \cite{mathKT150906271} or the description of $\KR$-groups by Boersema--Loring~\cite{mathOA150403284}. 
Moreover, our bulk index is the same thing as the index pairing in \cite{GrossmannSchulz-Baldes} (Remark \ref{remk:CTvsR}).
In particular, for type A and {\typeAII} topological insulators, it coincides with the Hall-conductance and the Kane-Mele invariant respectively (Example \ref{exmp:2dIQHE} and Example \ref{exmp:2dSQHE}). 
Even in this case, our approach gives a generalization of the bulk-edge correspondence for materials whose distribution of atoms have no translation symmetry such as quasi-crystals.
In particular, we obtain the bulk-edge correspondence for systems whose edge is not straight, which is studied from a different point of view for type A and {\typeAII} insulators by Prodan~\citelist{\cite{MR2554445}\cite{MR2525473}}.

Moreover, in this framework we can deal with topological crystalline insulators with symmetry-preserving disorder. 
For example, when we consider the group $\mathscr{A} \times \Zt$ where $\Zt$ acts on $\Euc^d$ as a reflection, we obtain the same classification list as indicated in \citelist{\cite{PhysRevB.88.075142}\cite{PhysRevB.88.125129}} (Example \ref{exmp:cryst}). 
Here, the use of twisted equivariant $\K$-theory is essential since we deal with the case that the reflection operator commutes or anticommutes with CPT operators.

It is pointed out by Jean Bellissard and Johannes Kellendonk that, although the uniform Roe algebra is similar to the groupoid $\Cst$-algebra of Bellissard, the (uniform) Roe algebras are not able to play the role of the noncommutative Brillouin zone since they are too large as the $\Cst$-algebra of observables. Actually, they are not even separable. 
A reasonable interpretation is that the uniform Roe algebra of $X$ is the universal algebra containing every observable algebras controlled at $X$.

\section{Preliminaries}\label{section:2}
\subsection{Twisted equivariant $\K$-theory}
Throughout this paper, we use the letter $\real$ for the real number field as a real $\Cst$-algebra (instead of the real line as a topological space) unless otherwise noted. 
We use the same letter for a real $\Cst$-algebra (a real Banach $\ast$-algebra satisfying the $\Cst$-condition) and the corresponding Real $\Cst$-algebra (a $\Cst$-algebra with an antilinear involution). For a complex $\Cst$-algebra $A$, we write $A_\real$ for the underlying real $\Cst$-algebra.
For an element $a$ in a Real $\Cst$-algebra $A$, we write $\lu 0 a :=a $ and $\lu 1 a:=\overline{a}$. 
\subsubsection{Definitions}
First, we prepare terminologies on twists of groups in the sense of \cite{MR3119923}. For more details, see for example Section 2 of \cite{Kubota1}.
\begin{defn}\label{defn:tw}
For a countable discrete group $G$, a \textit{twist} on $G$ is a triple $(\phi,c,\tau)$ where $\phi,c \in \Hom (G,\Zt)$ and $\tau \in H^2(G;\mathbb{T}_\phi)$, where $\mathbb{T}_\phi$ is the abelian group $\mathbb{T}$ with the $G$-action given by $g \cdot t:=\lu{\phi (g)}t$. 
\end{defn}
A $2$-cocycle $\tau \in H^2(G;\mathbb{T}_\phi)$ corresponds to a $\phi$-twisted central extension of $G$ by $\mathbb{T}$, that is, an extension
$1 \to \mathbb{T} \to G^\tau \to G \to 1$
such that $gtg^{-1}=\lu{\phi(g)} t $ for any $t \in \mathbb{T}$ and $g \in G^\tau$. 
For a twist $(\phi,c,\tau)$ of $G$, a $(\phi,c,\tau)$-twisted representation $(\Hilb, U)$ of $G$ (or a PUA representation \cite{mathph14067366}) is a unitary representation of $G^\tau$ on a $\Zt$-graded Hilbert space $\Hilb$ which is 
\begin{itemize}
\item $\phi$-linear ($U_g$ is complex linear if $\phi (g)=0$ and antilinear if $\phi (g)=1$),
\item $c$-graded ($U_g$ is even if $c(g)=0$ and odd if $c(g)=1$) and
\item $\tau$-projective ($U_t$ are the scalar multiplications for $t \in \mathbb{T}$). 
\end{itemize}
We write $\gamma _{\Hilb}$ for the grading operator of $\Hilb$. 
When $G$ is finite, we remark that $\lu \phi \Hilb _G^{c,\tau}:= (\ell^2(G) \oplus \ell^2(G)^{\mathrm{op}})^\infty \hotimes \mathscr{V}$ is independent of the choice of $(\phi,c,\tau)$-twisted representations $\mathscr{V}$ and contains all separable $(\phi,c,\tau)$-twisted representations of $G$ when $G$ is finite. 
We simply write $\lu \phi \Kop _G^{c,\tau}:=\Kop (\lu \phi \Hilb _G^{c,\tau})$. Note that $\Kop _{G}^{c_1,\tau _1} \hotimes \Kop _G^{c_2,\tau_2} \cong \Kop _G^{c,\tau}$ where $c:=c_1+c_2$ and $\tau:=\tau _1 +\tau _2 + \epsilon (c_1,c_2)$. Here, the $2$-cocycle $\epsilon (c_1,c_2)$ is determined by $\epsilon (c_1,c_2)(g,h):=(-1)^{c_1(h)c_2(g)}$.

A \textit{$\phi$-twisted $G$-$\Cst$-algebra} is a $\Cst$-algebra $A$ together with a $\phi$-linear $G$-action $\alpha : G\times A \to A$. For example, any Real $G$-$\Cst$-algebra has a structure of $\phi$-twisted $\Cst$-algebra for arbitrary $\phi$ by the action $\tilde{\alpha} _g(a):=\lu{\phi(g)}\alpha _g(a)$. In particular, we call the action $\alpha _g(a):=\lu{\phi(g)}a$ the \emph{trivial} $\phi$-twisted action of $G$ on a Real $\Cst$-algebra $A$.

For a $\phi$-twisted $G$-$\Cst$-algebra $A$, we say that $a \in A$ is \emph{$c$-twisted $G$-invariant} if $\alpha _g(a)=(-1)^{c(g)}a$ and we write $\mathscr{F}^G_{c,\mathscr{V}}(A)$ for the space of $c$-twisted $G$-invariant symmetries (i.e.\ self-adjoint unitaries) $s \in A \hotimes \Kop (\mathscr{V})$. For an inclusion $\mathscr{V} \subset \mathscr{W}$, we have $\mathscr{F}^G_{c,\mathscr{V}}(A) \to \mathscr{F}^G_{c, \mathscr{W}}(A)$ given by $s \mapsto s \oplus \gamma_{\mathscr{V}^\perp}$.

\begin{defn}[Theorem 5.14 of \cite{Kubota1}]\label{defn:twK}
Let $G$ be a finite group and let $A$ be a $\phi$-twisted $G$-$\Cst$-algebra. We define the twisted equivariant $\K _0$-groups by
\[
\lu \phi \K _{0,c,\tau}^G(A):= \bigcup \nolimits_{\mathscr{V}} \mathscr{F}^G_{c,\mathscr{V}}(A)/\sim _{\mathrm{homotopy}}
\]
where $\mathscr{V}$ runs over all finite dimensional $(\phi,c,\tau)$-twisted representations.
\end{defn}
In the same way as the usual $\K$-theory, we have an isomorphism
\[ \lu \phi \K _{-1,c,\tau}^G(A)\cong \lu \phi \K _{0,c,\tau }^G (\Bop (\Hilb _{G,A})/ \Kop (\Hilb _{G,A}))\]
where $\Hilb _{G,A}:=\ell^2(G)^\infty \otimes A$.
In other words,
\[ \lu \phi \K _{-1,c,\tau}^G(A) \cong \mathrm{Fred}_{c}^{G} (\lu \phi \Hilb _{G,A}^{c,\tau})/\sim _{\mathrm{homotopy}},\]
where $\mathrm{Fred}_{c}^G (\lu \phi \Hilb _{G,A}^{c,\tau})$ is the set of bounded, $c$-twisted $G$-invariant and self-adjoint operators $F$ on $\lu \phi \Hilb _{G,A}^{c,\tau}:=\lu \phi \Hilb _{G}^{c,\tau} \otimes A$ such that $F^2-1$ is compact.
For an exact sequence $0 \to I \to A \to A/I \to 0$ of $\phi$-twisted $G$-$\Cst$-algebras, we define the homomorphism
\[
\partial : \lu \phi \K _{0,\tau,c}^G(A/I)\to \lu \phi \K _{-1,\tau,c}^G(I), \ \partial [s]=[\sigma(\hat{s})]
\]
where 
\[ \sigma: A \hotimes \Kop(\mathscr{V}) \subset A \hotimes \Kop (\lu \phi \Hilb _G^{c,\tau})  \to \mathcal{M}(I \otimes \Kop (\lu \phi \Hilb _{G}^{c,\tau})) \cong \Bop (\lu \phi \Hilb _{G,I}^{c,\tau})\]
is the $\ast$-homomorphism as in Proposition 2.1 of \cite{MR1325694} (for a $\Cst$-algebra $A$, we write $\mathcal{M}(A)$ for its multiplier algebra) and $\hat{s}$ is a $c$-twisted $G$-invariant self-adjoint lift of $s$ in $A \hotimes \Kop(\mathscr{V})$. Theorem 5.14 of \cite{Kubota1} asserts that 
\[
\cdots \to \lu \phi \K _{0,\tau,c}^G(A) \to \lu \phi \K _{0,\tau,c}^G(A/I) \xra{\partial} \lu \phi \K _{-1,\tau,c}^G(I) \to \lu \phi \K _{-1,\tau,c}^G(A) \to \cdots 
\]
is a long exact sequence.

\begin{remk}\label{remk:detw}
In \cite{Kubota1}, the twisted equivariant $\K$-groups $\lu \phi \K ^G_{*,c,\tau}(A)$ are in general defined for $\Zt$-graded $\phi$-twisted $G$-$\Cst$-algebras. 
We remark that 
\maa{\lu \phi \K^G_{*,c,\tau}(A) \cong \lu \phi \K_*^G(A \hotimes \Kop_G^{c,\overline{\tau}}) \cong \KR^G_*((A \hotimes \Kop_G^{c,\overline{\tau}})_\real) \label{form:detw}}
where $\overline{\tau}:=\tau + \epsilon (c,c)$ (Proposition 3.12 of \cite{Kubota1}). 
\end{remk}

\begin{remk}
For a $\phi$-twisted $\Zt$-graded $G$-$\Cst$-algebra $A$, the twisted crossed product $G \ltimes ^\phi _{c,\tau}A$ is the $\Zt$-graded real $\Cst$-algebra generated by unitaries $\{ u_g \}_{g \in G}$ and $A_\real$ with the relations 
\[ u_gau_g^*=(-1)^{c(g)|a|}\alpha _g(a), u_gu_h=\tau(g,h)u_{gh}, \]
which has a canonical $\Zt$-grading determined from that of $A$ and $\gamma(u_g):=(-1)^{c(g)}u_g$ (for more precise definition, see Definition 4.4 of \cite{Kubota1}). The twisted Green-Julg theorem (Theorem 4.10 of \cite{Kubota1}) asserts that there is an isomorphism
\maa{\lu \phi \K ^G_{*,c,\tau}(A) \cong \KR_*(G \ltimes^\phi_{c,-\overline{\tau}} A ). \label{form:GJ}}
\end{remk}

\begin{remk}\label{remk:KK}
In Section 3 of \cite{Kubota1}, the twisted analogue of the $\KK$-group, in particular the $\K$-homology group, is also defined. 

An element of $\lu \phi \KK^G_{c,\tau}(A,\real)$ is represented by a \emph{$\lu \phi \K_G^{0,c,\tau}$-cycle} , that is, a triple $(\Hilb, \varphi, F)$ where $\Hilb$ is a $(\phi,c,\tau)$-twisted representation of $G$, $\varphi : A \to \Bop (\Hilb)$ is a graded $\ast$-homomorphism, and $F \in \Bop (\Hilb)$ is an odd self-adjoint operator such that $[F,\varphi (a)], \varphi (a)(F^2-1), \varphi (a)[U_g,F] \in \Kop$ and $[U_g,\varphi (a)]=0$ for any $a \in A$ and $g \in G$. 

The twisted Kasparov product gives a pairing 
\[\langle \blank , \blank \rangle : \lu \phi \K ^G_{c_1,\tau_1} (A) \otimes \lu \phi \KK ^G_{c_2,\tau_2}(A, \real) \to \lu \phi \K ^G_{c,\tau}(\real), \]
where $c=c_1+c_2$ and $\tau =\tau_1+\tau_2 +\epsilon (c_1,c_2)$.
\end{remk}

\begin{remk}\label{remk:lind}
For a subgroup $H \leq G$, there is a canonical isomorphism
\[ \lu \phi \K^G_{*,c,\tau}(A \otimes C(G/H)) \cong \lu \phi \K^H_{*,c,\tau}(A)  \]
given by $[s] \mapsto [s_{eH}]$, where $s=\bigoplus s_{gH} \in (A\hotimes \Kop (\mathscr{V})) \otimes C(G/H)$ (Theorem 20.5.5 of \cite{MR1656031}). Hence we obtain the \emph{$\ell ^2$-induction} functor 
\[ \lInd _H^G: \lu \phi \K^H_{*,c,\tau}(A) \cong \lu \phi \K^G_{*,c,\tau}(A \otimes C(G/H)) \xra{l_*} \lu \phi \K^G_{*,c,\tau}(A),\]
where $l:C(G/H) \to \Kop (\ell^2(G/H))$ is the inclusion. 

For example, in the case of topological $\K$-theory, $\Ind _G^{G \times \mathscr{T}} : \K_G^0(X) \to \KR_G^0(X)$ maps a complex vector bundle $V$ to the underlying real vector bundle $V_\real$.

Passing through the Green--Julg isomorphism (\ref{form:GJ}), $\lInd _H^G$ is the same thing as the homomorphism of $\KR$-groups induced from the inclusion of real $\Cst$-algebras $H \ltimes ^\phi _{c,-\overline{\tau}}A \to G \ltimes ^\phi_{c,-\overline{\tau}}A$. 
\end{remk}

Finally, we remark that the notion of twists and twisted equivariant $\K$-theory are in general defined for groupoids. For the definition, see Section 7 of \cite{MR3119923}. In particular, for an action groupoid $G \ltimes X$, a twist $\nu \in H^2_G(X, \mathbb{T}_\phi)$ corresponds to a $(\phi ,c,\tau)$-twisted $G$-equivariant Hilbert bundle $\mathscr{E}$ (Theorem 2.1 of \cite{Kubota1}) and $\lu \phi \K^{*,c,\nu}_G(X) \cong \lu \phi \K ^{*}_G(\Kop(\mathscr{E}))$.

\begin{exmp}\label{exmp:ext}
Let $1 \to\Pi \to G \to Q \to 1$ be an extension of discrete groups such that $Q$ is finite and $\Pi$ is free abelian and let $(\phi,c,\tau)$ be a twist of $Q$. 
We write $B_\Pi$ for the Pontrjagin dual of $\Pi$, on which $q \in Q$ acts as 
\[ (q \cdot \chi)(a)=\lu{\phi(q)}\chi (\tilde{q}a\tilde{q}^{-1})\]
where $\tilde{q}$ is a lift of $q$ on $G$. 
As in Section 2 of \cite{MR1909514}, we write $\ell ^2_\Pi(G)$ for the completion of $c_c(G)$ by the $C^*_r(\Pi)$-valued inner product
\[ \langle f_1,f_2 \rangle= \sum\nolimits_{h \in \Pi} \overline{f_1(gh)}f_2(h) \lambda _h. \]
In other words, $\ell^2_\Pi(G)$ is the section space of a Hilbert bundle over $B_\Pi$. Hence the $\phi$-twisted $G$-action $U_g(f)(h):=\lu{\phi(g)} f(gh)$ on $\ell^2_\Pi(G)$ gives a $(\phi,\nu)$-twisted action of $Q$ for some $\nu$. 
By Theorem 4.1 and Corollary 3.7 of \cite{MR1002543}, we can check that 
\[G \ltimes B \cong Q \ltimes _\nu (C^*_r(\Pi) \otimes B) \sim_{\mathrm{Morita}} Q \ltimes (\Kop (\ell^2_\Pi(G)) \otimes B) \]
for any $Q$-$\Cst$-algebra $B$ and hence we obtain
\maa{\lu \phi \K_Q^{*,c,\tau+\nu}(B_\Pi) \cong \lu \phi \K^Q_{*,c,\tau}(\Kop(\ell ^2_\Pi(G)))  \cong \KR_*(G \ltimes (\lu \phi \Kop_Q^{c,\overline{\tau}})_\real) \label{form:extK}}
by (\ref{form:detw}) and the Green--Julg theorem.

When $\nu =0$ (that is, $G \cong Q \ltimes \Pi$), the collapsing map $B_\Pi \to \pt$ and the inclusion $\pt \to B_\Pi$ into the origin induce a direct sum decomposition 
\[ \lu \phi \K _Q^{*,c,\tau}(B_\Pi) \cong \lu \phi \K^{*,c,\tau} _Q(B_\Pi, \pt) \oplus \lu \phi \K^{*,c,\tau} _Q(\pt),\]
where $\lu \phi \K^{*,c,\tau} _Q(B_\Pi, \pt)$ is the reduced $\K$-group and $\lu \phi \K^{*,c,\tau} _Q(\pt)$ is the trivial part. 
On the other hand, there is no such decomposition when $\nu \neq 0$. 
In this case, we characterize ``trivial'' elements as following. 
For each finite subgroup $F$ of $G$, we identify it with the image $q(F) \subset Q$ and set $H:=q^{-1}(F)$. 
Then, $H \cong F \ltimes \Pi$ since there is a splitting. Now we consider $\lInd _F^Q: \lu \phi \K^{*,c,\tau} _F(B_\Pi) \to \lu \phi \K^{*,c,\tau+\nu} _Q(B_\Pi)$ and we write 
\[ \mathrm{Triv} := \mathop{\mathrm{span}} \{ \lInd _{F}^Q (\lu \phi \K^{*,c,\tau} _F(\pt)) \mid \text{$F \leq G$ is finite} \} \subset \lu \phi \K^{*,c,\tau +\nu}_Q(B_\Pi). \]
\end{exmp}

\subsubsection{CT-type symmetries}
For a quadruple $(G, \phi , c, \tau)$, we write $G_0:=\Ker (\phi,c)$ and $\mathscr{A}:=\Im (\phi,c) \subset \Zt \times \Zt$.
\begin{defn}
We say that a quadruple $(G, \phi,c,\tau)$ is
\begin{itemize}
\item a \emph{CT-type symmetry} if there is a splitting $s: \mathscr{A} \to G$ and $\tau$ is the pull-back of a $2$-cocycle on $\mathscr{A}$,
\item a \emph{product CT-type symmetry} if $G=\mathscr{A} \times G_0$ and $\tau$ is a $2$-cocycle on $\mathscr{A}$, 
\item a \emph{CT-symmetry} if $G_0$ is trivial.
\end{itemize}
For a CT-type symmetry, we say that the pair $(\mathscr{A}, \tau)$ is its \textit{CT-type}.
\end{defn}

It is proved in Proposition 6.4 of \cite{MR3119923} that there are $10$ choices of pairs $(\mathscr{A}, \tau)$ where $\mathscr{A}$ of $\mathscr{G}:=\Zt \times \Zt$ and $\tau$ is a $2$-cocycle of $\mathscr{A}$, each of which corresponds to a Cartan labels as indicated in Table \ref{table:CT}. Here, we write $\mathscr{C}$, $\mathscr{T}$, $\mathscr{P}$ for subgroups of $\mathscr{G}$ generated by $\underline{C}:=(1,1)$, $\underline{T}:=(1,0)$, $\underline{P}:=(0,1)$ respectively. In fact, $(\mathscr{A},\tau)$ is classified by $C^2=\pm 1$ and $T^2=\pm 1$ where $C,T,P$ are lifts of $\underline{C}, \underline{T},\underline{P}$ in $\mathscr{A}^\tau$ such that $P^2=1$ and $CT=P$.

\begin{lem}[Corollary 5.16 of \cite{Kubota1}]\label{lem:KCT}
Let $(G \times \mathscr{A},\phi,c,\tau)$ be a product CT-type symmetry and let $A$ be a Real $G$-$\Cst$-algebra on which $\mathscr{A}$ acts trivially. Then, $\lu \phi \K_{0,c,\tau}^{G \times \mathscr{A}}(A)$ is isomorphic to one of $\K_*^G(A)$ or $\KR_*^G(A)$ as indicated in Table \ref{table:CT} below.
\end{lem}
It follows from the fact that the twisted group algebras of CT-symmetries are isomorphic to Clifford algebras and (\ref{form:GJ}). In this paper, we simply write as $\KL_* ^ G(A):=\lu \phi \K_{*,c,\tau}^{G \times \mathscr{A}}(A)$ by using its Cartan label ${\mathrm{L}}$. For example, $\KAII^G_*(A) \cong \KQ_*^G(A)$ and $\KBDI ^G_*(A) \cong \KR_{*+1}^G(A)$.

\begin{table}[h]
\scalebox{0.92}[0.92]{
\begin{tabular}{|r||c|c||c|c|c|c|c|c|c|c|} \hline
$\mathscr{A}$ & $1$ & $\mathscr{P}$ & \multicolumn{2}{c|}{$\mathscr{T}$} & \multicolumn{2}{c|}{$\mathscr{C}$} & \multicolumn{4}{c|}{$\mathscr{G}$}\\ \hline
$C^2$& \diaghead{$\Cliff _{4}^2$}{}{} & \diaghead{$\Cliff _{4,2}^2$}{}{} & \diaghead{$\Cliff _{4,2}^2$}{}{} & \diaghead{$\Cliff _{4,2}^2$}{}{} & $1$& $-1$ & $1$ & $1$ & $-1$ & $-1$ \\ \cline{1-3}\cline{4-11}
$T^2$& \diaghead{$\Cliff _{4}^2$}{}{} &\diaghead{$\Cliff _{4,2}^2$}{}{}& $1$ & $-1$ & \diaghead{$\Cliff _{4,2}^2$}{}{}  & \diaghead{$\Cliff _{4,2}^2$}{}{} & $1$ & $-1$ & $1$ & $-1$ \\ \hhline{|=#=|=#=|=|=|=|=|=|=|=|} 
Cartan &\typeA & \typeAIII  & \typeAI & \typeAII  & \typeD & \typeC & \typeBDI & \typeDIII  & \typeCI  & \typeCII \\ \hline
$\KL _0^G $ & $\K _{0}^G$ & $\K _{1}^G$ & $\KR _0^G$ & $\KR _{4}^G$ & $\KR _{2}^G$ & $\KR _{6}^G$ & $\KR _{1}^G$ & $\KR _{3}^G$ & $\KR_{7}^G$ & $\KR _{5}^G$ \\ \hline
\end{tabular}
}
\caption{The 10-fold way for twisted equivariant $\K$-groups}
\label{table:CT}
\end{table}

\begin{remk}\label{remk:CTvsR}
The groups $\KL_0(A)$ are essentially the same thing as the $DK$-groups studied in \cite{mathKT150906271}. Moreover, for real Cartan labels, the isomorphisms $\KL_0(A) \cong \KR _j(A)$ are given in another way by Boersema--Loring~\cite{mathOA150403284}. 
Actually, when $P \not \in \mathscr{A}$, $U^{(j)}_\infty (A)$ in \cite{mathOA150403284} is the same thing as $\bigcup \lu \phi \F_{c,\mathscr{V}}^{\mathscr{A}}(A)$. Similarly, when $P \in \mathscr{A}$, $s \mapsto \Phi(s\gamma_\mathscr{V})\Phi$ (where $\Phi:=(1+P)/2$) gives a one-to-one correspondence between $\bigcup \lu \phi \F_{c,\mathscr{V}}^{\mathscr{A}}(A)$ and $U^{(j)}_\infty (A)$. 
Set $\mathscr{V}':=\mathscr{V}$ if $P \not \in \mathscr{A}$ and $\mathscr{V}':=\Phi\mathscr{V}$ if $P \in \mathscr{A}$.

They introduce Real $\Cst$-algebras $A_j$ which are $\KK$-equivalent to $\Cl_{0,j}$ and $\KKR(A_j, A) \cong [A_j, A \otimes \Kop]$. 
By functoriality, for any $s \in \lu \phi \F_{c,\mathscr{V}}^{\mathscr{A}}(A)$ there is a $\ast$-homomorphism $\psi _s:A_j \to A \otimes \Kop(\mathscr{V}')$ such that $[s]=[\psi _s(u_0)]$ under the above identification, where $u _0\in U_\infty ^{(j)}(\tilde{A}_j)$ is a distinguished element such that $[u_0] \in \KR_j(A_j) \cong \zahl$ is the generator.

This identification is useful for studying the Kasparov product. Let $(\mathscr{A},\phi,c,\tau)$ be the CT-symmetry corresponding to $\KR_d$-theory and set $\Hilb _d:=\Hilb _{\mathscr{A}}^{c,\tau}$. Set
\[\mathrm{FS}_{j,d}:=\{ (F,s) \in \lu \phi \mathrm{Fred}_c^{\mathscr{A}}(\Hilb _d)^{\mathrm{odd}} \times U^{(j)}(\Bop (\Hilb _d))^{\mathrm{even}} \mid  [F,s] \in \Kop \}\]
if $j=-1,0,1,2$ and $\mathrm{FS}_{j,d}:=\mathrm{FS}_{j-4,d+4}$ if $j=3,4,5,6$.
Applying the above identification for the $\Cst$-algebra $D:=\{ T \in \Bop(\Hilb _d)^{\mathrm{even}} \mid [T,F] \in \Kop \}$, we obtain an isomorphism $\KR^d(A_j) \cong \pi_0(\mathrm{FS}_{j,d})$ (note that $\KR^d(A_{j+4})=\KR^d(A_j \otimes \mathbb{H}) \cong \KR^{d+4}(A_j)$). 

Let $s \in \lu \phi \F^\mathscr{A}_{c,\mathscr{V}}(A)$ corresponding to $\psi _s : A_j \to A \otimes \Kop (\mathscr{V}')$. Then, by the above identification, the Kasparov product
\[\KKR(A_j,A) \otimes \KKR^d(A,\real) \to \KKR^d(A_j,\real)\]
is given by 
\[[\psi _s] \otimes_A [\Hilb , \varphi , F] = [\Hilb \otimes \mathscr{V}', \psi _s \circ \varphi, F \otimes 1],\]
which corresponds to $[F \otimes 1, (\varphi \otimes \id_{\Kop (\mathscr{V}')})(s)] \in \pi_0(\mathrm{FS}_{j,d})$ under an arbitrary embedding of $\Hilb \otimes \mathscr{V}'$ into $\Hilb _d$ ($j=-1,0,1,2$) or $\Hilb _{d+4}$ ($j=3,4,5,6$).

Explicit isomorphisms of $\pi _0(\mathrm{FS}_{j,d})$ with $\zahl$ or $\Zt$ are given by the index pairings investigated by Gro{\ss}mann--Schulz-Baldes~\cite{GrossmannSchulz-Baldes}. In other words, their index pairing is the same thing as the Kasparov product.
\end{remk}

\begin{exmp}\label{exmp:rep}
The groups $\KR_j(\real)$ for $j=0,1,2,3,4,5,6,7$ are isomorphic to $\zahl, \zahl _2, \zahl _2, 0 , \zahl , 0,0,0$ (the $8$-fold Bott periodicity).
Hence the twisted equivariant $\K$-group $\lu \phi \K ^{\mathscr{A}}_{0,c,\tau}(\real)$ for CT-symmetries are isomorphic to one of $\zahl $, $\Zt$ or $0$ as indicated in Table \ref{table:Kitaev}. 
\end{exmp}
\begin{table}[h]
\scalebox{0.90}[0.90]{
\begin{tabular}{|r||c|c||c|c|c|c|c|c|c|c|} \hline
$d$ & A          & {\typeAIII}    & {\typeAI}         &{\typeBDI}       & D            & {\typeDIII}       & {\typeAII}          & {\typeCII}        & C            & {\typeCI}      \\ \hline
$0$      & $\zahl$ & $0$      & $\zahl$ & $\Zt$    & $\Zt $     & $0$        & $\zahl $ & $0$        & $0$        & $0$     \\ \hline
$1$      & $0$       &$\zahl$ & $0$       & $\zahl$ & $\Zt$      & $\Zt $    & $0$        & $\zahl $ & $0$        &  $0$    \\ \hline
$2$      & $\zahl$ & $0$      & $0$       &$0$        & $\zahl$   & $\Zt$     & $\Zt $    & $0$        & $\zahl $ & $0$       \\ \hline 
$3$      & $0$      & $\zahl$       &$0$        &$0$        &$0$    & $\zahl$  & $\Zt$     & $\Zt $    & $0$        & $\zahl $   \\ \hline
\end{tabular}
}
\caption{the group $\KL_d(\real )$}
\label{table:Kitaev}
\end{table}

\begin{exmp}
The group $\KR^{G}_{d}(\real)$ for $d=0,\dots, 7$ is written in terms of the representation ring as
\[ R_\real (G), R_\real (G)/\rho R(G), R(G)/jR_\mathbb{H}(G), 0, R_\mathbb{H}(G), R_{\mathbb{H}}(G)/\eta R(G), R(G)/iR_\real (G), 0 \]
(for the definition of $i$, $j$, $\rho$ and $\eta$, see Section 6 and Section 8 of \cite{MR0259946}). 
Hence, we can write down $\KL ^G_*(\real)$ in terms of the representation rings.

These groups are also understood in terms of the Frobenius--Schur indicator  (Chapter 13.2 of \cite{MR0450380}).
For any finite group $G$, 
\[ C^*_r(G) \cong \bigoplus _{\pi \in \Irr(G)_1} \Mop _{l_\pi}(\real) \oplus \bigoplus _{[\pi] \in \Irr(G)_0/\sim} \Mop_{l_\pi}(\comp) \oplus \bigoplus _{\pi \in \Irr(G)_{-1}} \Mop _{l_\pi/2}(\quot),  \]
where $l_\pi:=\dim_\comp \pi$, $\Irr(G)_{i}$ is the set of irreducible representations whose Frobenius--Schur indicator are $i=1,0,-1$ and $\pi \sim \pi'$ in $\Irr(G)_0$ if and only if $\pi'=\pi$ or $\pi^*$. Hence we obtain
\maa{\KR_{*}^G(\real) \cong  \KR_{*}(\real)^{n_1} \oplus \K_{*}(\comp)^{n_0/2} \oplus \KQ_{*}(\real)^{n_{-1}} \label{form:group}}
where $n_i:=|\Irr (G)_i|$ for $i=1,0,-1$.

For example, consider the case that $G=\zahl /k\zahl$. By the Pontrjagin duality, 
we obtain $C^*_r(C_{k}) \cong C(\hat{G})$ with the involution $\chi \mapsto \overline{\chi}$ (here $\hat{G}$ is identified with the group of roots of unity). Hence we obtain
\maa{ 
\KL^{C_k}_*(\real ) \cong \begin{cases}
\K_{j+*}(\comp)^k & \KL_* \cong \K_{j+*}, \\
\KR_{j+*} ( \real)  \oplus \K_{j+*} (\comp) ^{\frac{k-1}{2}} & \text{$\KL_* \cong \KR_{j+*}$, $k$ is odd}\\
\KR_{j+*} (\real) ^2 \oplus \K_{j+*}(\comp) ^{\frac{k-2}{2}} & \text{$\KL_* \cong \KR_{j+*}$, $k$ is even}.
\end{cases}
\label{form:rot}}
\end{exmp}

\begin{exmp}\label{exmp:ind}
As in Example \ref{exmp:rep}, $\KAII ^2_G(\real )$, $\KDIII ^1_G(\real )$ and $\KCII ^3_G(\real )$ are isomorphic to $R(G)/R_{\mathbb{H}}(G)$. 
In fact, the canonical surjection from $R(G)$ to these groups is given by the $\ell ^2$-induction from $G':=\Ker \phi$. 

To see this, first we reduce the problem to $\lInd _e^{\mathscr{T}}: \K_0^G(\Cliff_2) \to \KR_0^G(\Cl_{2,0})$ by the Green--Julg isomorphism (\ref{form:GJ}) and show that $\lInd _e^{\mathscr{T}}:\K^{-2}_G(\pt) \to \KR^{-2}_G(\pt )$ is the desired surjection. 
It follows from the fact that  the surjection $R(G) \cong M^\real_2(G) \to \KR_G^0(D^2,S^1) $ in Section 8 of \cite{MR0259946} maps $[V]$ to
\[[ D ^2 \times V_\real , D^2 \times V_\real , x+iy ] \in \KR _G^0 (D^2,S^1),\]
which is the same thing as $\lInd _e^{\mathscr{T}} (\beta \otimes [V])$ where $\beta$ is the Bott generator.
\end{exmp}

\begin{exmp}\label{exmp:reflect}
Let $\mathscr{R}:=\Zt$ and consider an arbitrary twist on $G:=\mathscr{R} \times \mathscr{A}$ such that $\phi$ and $c$ come from that of $\mathscr{A}$. We write the generator of $\Zt$ as $\underline{R}$. 
In the same way as Lemma 6.17 of \cite{MR3119923}, lifts $C,T,R$ of $\underline{C}$, $\underline{T}$ and $\underline{R}$ satisfies $(CT)^2=1$ and $R^2=1$. Hence, $\tau$ is classified by $C^2=\pm 1$, $T^2=\pm 1$, $CR=\pm RC$ and $TR=\pm RT$. 
The classification list is given in Table I of \cite{PhysRevB.88.075142}.

Let us consider the $\mathscr{R}$-action on $\Cl _{0,1}$ induced from the reflection of $\real$, that is, the $\ast$-automorphism given by $\underline{R}\cdot 1=1$ and $\underline{R} \cdot e=-e$. 
For the latter use, we give a calculation of the twisted equivariant $\K$-group $\lu \phi \K ^G_{0,c,\tau}(\Cl _{0,1})$. 

First, we consider the case that $P \not \in \mathscr{A}$ or $PR=RP$. 
By definition of the twisted crossed product, we have the isomorphism
\[ G \ltimes ^\phi _{c,\overline{\tau}} \Cl_{0,1} \cong \mathscr{A} \ltimes ^\phi _{c,\overline{\tau}} (\mathscr{R} \ltimes \Cl _{0,1}). \]
Here, $\mathscr{A}$ acts on $\mathscr{R} \ltimes \Cl_{0,1}$ as $e_1 \mapsto e_1$ and $R \mapsto \eta(g)R$, where $\eta(g) \in \{\pm 1\}$ is determined by $Rg=\eta(g)gR$.

Now, the underlying complex $\Cst$-algebra $\mathscr{R} \ltimes \Cliff_1$ is isomorphic to $\Cliff_{2}$ generated by $e$ and $iRe$.
Hence the fixed point subalgebra of the $\mathscr{A}$-action on it is $\Cl_{1,1}$ when $TR=RT$ and $\Cl_{0,2}$ when $TR=-RT$.  
In other words, the $\mathscr{A}$-action on $R \ltimes \Cl_{0,1}$ is the same thing as the trivial action on $\Cl_{1,1}$ (resp.\ $\Cl_{0,2}$) when $TR=RT$ (resp.\ $TR=-RT$).

For the case that $PR=-RP$, apply the same argument for $G=\mathscr{A} \times \mathscr{R}'$, where $\mathscr{R}'$ is the subgroup generated by $RP$. 
It is directly checked that the underlying complex $\Cst$-algebra $\mathscr{R}' \ltimes_c \Cliff_1$ is isomorphic to $\Cliff _1 \oplus \Cliff_1$ (generated by $(1+RP)e/2$ and $(1-RP)e/2$) and the $\mathscr{A}$-action on it is the same thing as the trivial action on $\Cl_{0,1} \oplus \Cl_{0,1}$ (resp.\ $(\Cliff _1)_\real$)when $TRP=RPT$ (resp.\ $TRP=-RPT$).

Finally, by the Green-Julg isomorphism (\ref{form:GJ}) we obtain 
\[
\lu \phi \K ^G_{d-1,c,\tau}(\Cl_{0,1}) \cong \begin{cases}\lu \phi \K^{\mathscr{A}}_{d-1,c,\tau}(\real) & \text{if $PR=RP$ and $TR=RT$},\\
\lu \phi \K^{\mathscr{A}}_{d+1,c,\tau}(\real) & \text{if $PR=RP$ and $TR=-RT$}.\\
\lu \phi \K^{\mathscr{A}}_{d,c,\tau}(\real)^2 & \text{if $PR=-RP$ and $TRP=RPT$},\\
\K_{d,c,\tau}(\real) & \text{if $PR=-RP$ and $TRP=-RPT$}.\\
\end{cases}
\]
\end{exmp}

\subsection{Coarse $\Cst$-algebras}
Next, we review definitions and basic properties of $\Cst$-algebras arising in coarse geometry and set up our notation. The reader can find more details in \citelist{\cite{MR1147350}\cite{MR1399087}\cite{MR1817560}\cite{MR2007488}}.
Let $X$ be a proper (i.e.\ all bounded closed subsets are compact) discrete metric space with a proper isometric action of $G$. 

\begin{defn}
A triplet $(\Hilb , U, \pi)$ is a \textit{(Real) $(X,G)$-module} if $(\Hilb,U)$ is a (Real) unitary representation of $G$ and $\pi:c_0(X) \to \Bop (\Hilb)$ is a $G$-equivariant (Real) $\ast$-homomorphism (equivalently, a decomposition $\Hilb =\bigoplus _{x \in X} \Hilb _x$ compatible with the $G$-action). We say that a (Real) $(X,G)$-module is \textit{finite} or \textit{uniformly finite} if the rank of projections $\pi (\delta _x)$ ($x \in X$) are finite or uniformly finite respectively.
\end{defn}

Hereafter, we make a technical assumption on $G$: the order of finite subgroups of $G$ is uniformly bounded. 
For example, any discrete group acting on $\Euc ^d$ satisfies this condition. 
For such $G$ acting properly and isometrically on $X$,  set $\tilde{X}:=\{ (g,x) \in G \times X \mid gx=x \}$, which is a free $G$-space by the $G$-action $g \cdot (h,x):=(hg^{-1}, gx)$.
Then, we obtain a uniformly finite Real $(X,G)$-module 
\[\Hilb_{X,G}:= \ell ^2 (\tilde{X})=\bigoplus\nolimits_{x\in X} \ell ^2 (G_x) \]
with the Real structure given by $\overline{f}(x):=\overline{f(x)}$. 

\begin{defn}\label{defn:Roe}
Let $X$, $G$ be as above and let $(\Hilb, U, \pi)$ be a (Real) $(X,G)$-module. We say that $T \in \Bop(\Hilb)$ is 
\begin{itemize}
\item \emph{controlled} (or $T$ has \textit{finite propagation}) if there is $R>0$ such that $T_{xy}:=\pi (\delta_y)T\pi(\delta_x)=0$ if $d(x,y) >R$, 
\item \emph{locally compact} if $\pi (\delta _x)T$ and $T\pi(\delta _x)$ are compact operators for any $x \in X$,
\item \emph{supported near $A$} if there is $R>0$ such that $T_{xy}=0$ if $d(x , A)>R$ or $d(y,A)>R$.
\end{itemize}
We write $C^*_\Hilb(X)^G$ (resp.\ $C^*_\Hilb(A \subset X)^G$) for the closure of the (Real) $\ast$-algebra of $G$-invariant, controlled and locally compact operators on $\Hilb$ (resp.\ supported near $A$).
In particular, we simply write 
\[\begin{array}{l}
C^*_u(X)^G := C^*_{\Hilb _{X,G}}(X)^G,\\
C^*(X)^G := C^*_{\Hilb _{X,G}^\infty}(X)^G,
\end{array}
\ 
\begin{array}{l}
C^*_u(A \subset X)^G:=C^*_{\Hilb _{X,G}}(A \subset X)^G, \\
C^*(A \subset X)^G := C^*_{\Hilb _{X,G}^\infty}(A \subset X )^G, 
\end{array}\]
and call them the \emph{$G$-invariant (relative) uniform Roe algebra} and  \emph{$G$-invariant (relative) Roe algebra} respectively. 
\end{defn}
The $\Cst$-subalgebra $C^*_\Hilb (A \subset X)^G$ of $C^*_\Hilb (X)^G$ is an ideal. In this paper, we simply write 
\[C^*_\Hilb (X, A )^G:=C^*_{\Hilb} (X)^G/C^*_{\Hilb} (A \subset X)^G\]
and define $C^*(X,A)^G$, $C^*_u(X,A)^G$ in the same way.

\begin{exmp}\label{exmp:eqRoe}
When $G$ acts on $X$ cocompactly, $C^*_u(X)^G$ is isomorphic to $\Kop (\ell ^2_G(X))$ (Lemma 2.3 of \cite{MR1909514}). 
For example, for the Cayley metric space $|G|$ on which $G$ acts by the left multiplication, we have
\[C^*_u(|G|)^G \cong C^*_r(G) , \  C^*(|G|)^G \cong C^*_r(G) \otimes \Kop, \]
where $C^*_r(G)$ is the reduced group $\Cst$-algebra.
Similarly, for an extension $1 \to \Pi \to G \to Q \to 1$ as in Example \ref{exmp:ext},  the $Q$-$\Cst$-algebra $C^*_{\ell ^2 (G)}(|G|)^\Pi$ is $Q$-equivariantly isomorphic to $\Kop (\ell^2_\Pi (G))$. 
\end{exmp}

\begin{exmp}\label{exmp:groupoid}
It is known that these $\Cst$-algebras are related to certain groupoids when $G$ is trivial.
Let $\G_X$ be the translation groupoid defined in Section 3.2 of \cite{MR1905840}. 
Then, by the same argument as Lemma 4.4 of \cite{MR1905840}, we obtain an isomorphism
\[C^*_\Hilb (X) \cong \G_X \ltimes c_b(X; \Kop(\Hilb)) \]
where $c_b(X;\Kop (\Hilb)):=\prod _{x \in X} \Kop (\Hilb _x)$. 
For example, $C^*_u(X) \cong C^*_r (\mathcal{G}_X)$ and $C^*(X) \cong \mathcal{G}_X \ltimes c_b(X,\Kop)$.
In particular, for the Cayley metric space $|G|$, we have 
\[C^*_u(|G|) \cong G \ltimes c_b(G), C^*(|G|) \cong G \ltimes c_b(G ,\Kop).\]
\end{exmp}

Example \ref{exmp:groupoid} immediately implies that the algebra of observables $\zahl ^d \ltimes C(\Omega )$ introduced in \citelist{\cite{MR862832}\cite{MR1295473}} is a subalgebra of $C^*_u(X)$. 
The same holds for the observable algebra of aperiodic solids.
Let $X$ be a Delone subset of $\Euc^d$ and let $\Gamma _X$ be the groupoid of the transversal of the hull of $X$ (Definition 1.9 of \cite{MR1798994}). Note that the target space $\Gamma _X^0$ is regarded as a compactification of $X$.
\begin{lem}\label{lem:groupoid}
There is a surjective groupoid homomorphism $\G _X \to \Gamma _X$ which induces the inclusion $C^*_r(\Gamma _X) \subset C^*_u(X)$.
\end{lem}
\begin{proof}
Since $\Gamma _X^0$ is a compactification of $X$, we have a unique surjection $f: \beta X \to \Gamma _X^0$. 
Recall that the morphism set $\G _X^1$ is given by the union of all subspaces of $\beta (X \times X)$ which are closures of controlled (i.e.\ $d(x,y)$ is uniformly bounded) subspaces of $X \times X$. 
Let $Z$ be a controlled subset of $X \times X$. Then, the bounded function $v: Z \to \real ^d$ determined by $v(x,y) \cdot x=y$ extends to the function on $\beta Z=\overline{Z} \subset \beta (X \times X)$. Now we obtain the desired groupoid morphism $z \mapsto (f(x),v(z))$.
\end{proof}

Next, we prepare some notions on inclusions of metric spaces related to the tight-binding approximation of the position space. 
For a subspace $Y$ of $X$ and $R>0$, we write $\Pen (Y,R):=\{ x \in X \mid d(x,Y)<R\}$. 
\begin{defn}\label{defn:Delone}
A discrete subspace $X$ of a metric space $M$ is 
\begin{itemize}
\item \emph{uniformly discrete} if there is $r>0$ such that any closed ball of radius $r$ in $M$ contains at most one point in $X$,
\item \emph{relatively dense} if there is $R>0$ such that any closed ball of radius $R$ in $M$ contains at least one point in $X$,
\item a \emph{Delone subset} if $X$ is uniformly discrete and relatively dense.
\end{itemize}
For s closed subspace $N$ of $M$, we say that a pair of discrete subspaces $Z \subset Y$ of $M$ is a \emph{Delone pair} of $N \subset M$ if $Y$ is a Delone subset of $M$ and the Hausdorff distance $d(N,Z)$ is finite (in other words, $Z$ is a Delone subset of $\Pen (N,R)$ for some $R$). For a partition $M=M_+ \cup M_-$, $X =Y_+ \cup Y_-$ is a \emph{Delone partition} of $M=M_+ \cup M_-$ if $X$ is a Delone subset of $M$ and $d(Y_\pm , M_\pm) <\infty$, $d(Z,N)<\infty$ holds, where $N:=M_+ \cap M_-$ and $Z:=Y_+ \cap Y_-$.
\end{defn}
A discrete subspace $X$ is relatively dense in $M$ if and only if the inclusion is a coarse equivalence. 
If $M$ has bounded geometry, so does $X$ if and only if $X$ is uniformly discrete. 
If a partition $M=M_+ \cup M_-$ is $\omega$-excisive (i.e.\ for any $R>0$ there is $S>0$ such that $\Pen (Y_1,R) \cap \Pen (Y_2,R) \subset \Pen (Y_1 \cap Y_2,S))$, then so are its Delone partitions.

Finally, we remark on coarse invariance of uniform Roe algebras. We say that two $G$-spaces are \emph{$G$-equivariantly coarsely equivalent} if there are $G$-maps $f : \tilde{X} \to \tilde{Y}$ and $g: \tilde{Y} \to \tilde{X} $ such that there is a nondecreasing function $\rho : \real _{\geq 0} \to \real _{\geq 0}$ and $R>0$ satisfying
\begin{itemize}
\item $d_Y(f(x),f(x'))) \leq \rho (d_X(x,x'))$ and $d_X(g(y),g(y')) \leq \rho (d_Y(y,y'))$,
\item $d_X(gf(x),x)<R, d_Y(fg(y),y)<R$,
\end{itemize}
for any $x,x' \in \tilde{X}$ and $y,y' \in \tilde{Y}$, where $d_X$ and $d_Y$ are metric functions on $X$ and $Y$ lifted to $\tilde{X}$ and $\tilde{Y}$. 
For example, two $G$-invariant Delone subsets of a $G$-metric space $M$ are $G$-equivariantly coarsely equivalent.

\begin{lem}\label{lem:Morita}
Let $X$ and $Y$ be uniformly discrete $G$-metric spaces with bounded geometry which are $G$-equivariantly coarse equivalent. Then, for any normal subgroup $\Pi$ of $G$, $C^*_{\Hilb_{X,G}}(X)^\Pi$ and $C^*_{\Hilb _{Y,G}}(Y)^\Pi$ are $G/\Pi$-equivariantly Morita equivalent \cite{MR0419677}.
\end{lem}
\begin{proof}
The proof is given in the same way as Theorem 4 of \cite{MR2363428} by using $\tilde{X}$ and $\tilde{Y}$ instead of $X$ and $Y$.
\end{proof}

\subsection{Twisted equivariant $\K$-theory of coarse $\Cst$-algebras}
Our main subject in this paper is the topology of $c$-twisted $G$-invariant controlled symmetries on \emph{$(\phi ,c ,\tau)$-twisted $(X,G)$-modules} (a variation of $(X,G)$-modules using $(\phi,c,\tau)$-twisted representations as $(\Hilb, U)$). 
In order to reduce the problem to calculation of twisted equivariant $\K$-groups, we make the following assumption (note that equivariant $\K$-groups are not represented by $G$-invariant operators when $G$ is not compact).
\begin{assump}\label{assump:ext}
There is an extension $1 \to \Pi \to G \to Q \to 1$ such that $Q$ is finite and $(\phi,c,\tau)$ is the pull-back of a twist of $Q$.
\end{assump}
For example, CT-type symmetries satisfy this assumption. In addition to that, we often assume that $\Pi$ is free abelian. 

\begin{exmp}\label{exmp:extRoe}
Let $G:=G' \times N$, where $N$ is a finite group acting on $\Euc ^d$ trivially and $G'$ is a discrete subgroup of the Euclidean group $\Aut (\Euc ^d)$. 
By Theorem 1 and Theorem 2 of \cite{MR574501}, we obtain 
\begin{itemize}
\item a group extension $1 \to \Pi \to G' \to Q' \to 1$ such that $\Pi$ is a cocompact free abelian normal subgroup of rank $n$,
\item a decomposition $\Euc ^d \cong \Euc ^n \times \Euc ^{d-n}$ such that $G'$ acts on $\Euc ^n$ cocompactly and on $\Euc ^{d-n}$ as a $Q'$-action. 
\end{itemize}
We remark that the $Q'$-action on $\Euc ^n$ has a fixed point $x_0 \in \Euc$ and hence it is identified with the induced linear action on $V:=T_{x_0}\Euc ^n$ by the exponential map.
Set $Q:=Q' \times N$ and let $(\phi,c,\tau)$ be a twist of $Q$. Then, $(G,\phi,c,\tau)$ satisfies Assumption \ref{assump:ext}
\end{exmp}

Let $(G,\phi,c,\tau)$ be as Assumption \ref{assump:ext} and let $X$ be a proper metric space on which $G$ acts properly and isometrically. 
For a $(X,G)$-module $\Hilb$, $C^*_\Hilb (X)^\Pi$ has a canonical Real $Q$-$\Cst$-algebra structure. Hereafter, we simply write 
\[C^*_{u,G}(X)^\Pi:=C^*_{\Hilb _{X,G}}(X)^\Pi, \ C^*_G(X) ^\Pi:= C^*_{\Hilb _{X,G}^\infty}(X)^\Pi\]
respectively. 
Since any uniformly finite $(\phi,c,\tau)$-twisted $(X,G)$-module $\Hilb$ is contained in $\Hilb _{X,G} \hotimes \mathscr{V}$ for some finite dimensional $(\phi,c,\tau)$-twisted representation $\mathscr{V}$, a controlled $c$-twisted $G$-invariant symmetry on $\Hilb$ determines an element of the group $\lu \phi \K ^Q_{0,c,\tau}(C^*_{u,G}(X)^\Pi )$.

\subsubsection{Direct calculations}
First of all, we give a calculation of the $\K$-group of uniform Roe algebras and Roe algebras of $\zahl ^d$ by using the Pimsner-Voiculescu exact sequence~\cite{MR587369}. 
For example, since $\K _0(c_b(\zahl)) \cong B(\zahl)$ (where $B(\zahl) \subset \zahl ^\zahl$ is the group of bounded sequences of integers), $\K_0(c_b(\zahl , \Kop)) \cong \zahl ^\zahl$ and $\K_1(c_b(\zahl , \Kop)) \cong 0 \cong \K_1(c_b(\zahl , \Kop))$,
we have exact sequences
\ma{0 \to \K_{1}(C^*_u(|\zahl|)) \to B(\zahl) &\xra{\mathrm{shift}} B(\zahl) \to \K_0 (C^*_u(|\zahl|)) \to 0, \\
0 \to \K_{1}(C^*(|\zahl|)) \to \zahl ^\zahl &\xra{\mathrm{shift}} \zahl ^\zahl \to \K_0 (C^*(|\zahl|)) \to 0.
}
In particular, $\K_0(C^*_u(|\zahl|)) =0$ although $\K_0(C^*_u(|\zahl|))$ has uncountably many basis (Example 3.4 of \cite{Spakula}).

By the same argument, we can calculate several $\KR_*$-groups of uniform Roe algebras and compare them with that of Roe algebras.
\begin{lem}\label{lem:uRoe}
The canonical map
\[ \KR_j(C^*_u(|\zahl ^d|)) \to \KR_j(C^*(|\zahl ^d|)) \]
is an isomorphism for $j=1,2,3,5,6,7$ when $d=1$, $j=1,2,5,6$ when $d=2$ and $j=1,5$ when $d=3$. Moreover, for $j=0,4$,
\[ \KR_j(C^*_u(|\zahl ^d|))/\Im (\iota_u) _* \to \KR_j(C^*(|\zahl ^d|)), \]
where $\iota _u : c_b(\zahl ^d) \to C^*_u(|\zahl ^d|)$ is the inclusion, is an isomorphism when $d=0,1,2,3$. 
\end{lem}
\begin{proof}
They are direct consequences of the Pimsner-Voiculescu exact sequence \cite{MR587369} for the crossed product $C^*_u(|\zahl^d|) \cong \zahl ^d \ltimes c_b(\zahl ^d)$. Note that $\KR_*(c_b(\zahl ^d)) \cong (\Zt)^{\zahl ^d}$ for $\ast=1,2$ and hence the shift homomorphisms are surjective.
\end{proof}

\subsubsection{Coarse Baum-Connes isomorphism}
Let $M$ be a uniformly contractible metric space with a proper isometric $G$-action and let $X$ be a $G$-invariant Delone subset of $M$. 
Following to \cite{MR1451759} and \cite{MR2565716}, we consider the \emph{twisted equivariant coarse assembly map} 
\[\mu_{X,G}: \lu \phi \K _{*,c,\tau}^G(M) \to \lu \phi \K _{*,c,\tau}^Q(C^*_G(X)^\Pi) \]
defined as following. Let $[\Hilb, \pi, F]$ be a $\lu \phi \K _{*,c,\tau}^G$-cycle of $C_0(M)$.
As in Lemma 9 of \cite{MR2565716}, we can replace $F$ with $F'$ which is $G$-invariant and controlled.
Moreover, as in Proposition 6.3.12 of \cite{MR1817560}, we regard $\Hilb$ as a $(\phi,c,\tau)$-twisted $(X,G)$-module by the $\ast$-representation $\pi'$ given by $\pi' (\delta _x):= \pi (\chi _{F_x})$, where $M=\bigcup _{x \in X} F_x$ is a uniformly bounded Borel partition such that $F_x \cap X=\{x\}$ and $g(F_x)=F_{gx}$. 

By taking direct sum with a degenerate cycle, we may assume that $\Hilb \cong \Hilb _{X,G}^\infty \hotimes \mathscr{V} \hotimes \hat{\real} ^{2,2}_c$ as $(\phi,c,\tau)$-twisted $(X,G)$-modules, where $\mathscr{V}$ is a finite dimensional $(\phi,c,\tau)$-twisted representation.
Therefore, we get
\[[F'] \in \lu \phi \K_{1,c,\tau}^Q(\mathcal{M}(C^*_G (X)^\Pi)/C^*_G(X)^\Pi)_\vD. \]
Here, for the definition of $\hat{\real} ^{2,2}_c$ and $\lu \phi \K_{1,c,\tau}^Q(\blank)_\vD$, see Definition 5.4 of \cite{Kubota1}. 
Finally, we define the coarse assembly map as
\[\mu_{X,G} [\Hilb , \pi , F] := \partial [F'] \in \lu \phi \K_{0,c,\tau}^Q(C^*_G(X)^\Pi). \]

\begin{thm}\label{thm:cBC}
When $M=\Euc ^d$, $\mu _{X,G}$ is an isomorphism. 
\end{thm}
\begin{proof}[Sketch of the proof]
In the proof, we use the Roe algebra for general (not necessarily discrete) metric spaces. For the definition, see for example Definition 6.3.8 of \cite{MR1817560}. We only remark that $C^*_G(X)^\Pi$ and $C^*_G(\Euc^d)^\Pi$ have the same (twisted equivariant) $\K$-theory (Proposition 6.3.12 of \cite{MR1817560}).

First, consider the case that the $G$-action on $\Euc ^d$ is cocompact. 
By (\ref{form:detw}), (\ref{form:extK}) and the same argument as in \cite{MR1909514}, $\mu_{X,G}$ is isomorphic to the Real Baum--Connes assembly map with coefficient in $\Kop(\mathscr{V})_\real$, which is an isomorphism since $G$ has the Haagerup property \cite{MR1487204}.

For general $G$, the proof is given in a similar way as in \cite{MR1451759}. Let us choose a decomposition $\Euc ^d \cong \Euc ^{d-n} \times \Euc ^n$ as in Example \ref{exmp:extRoe}.
For any $G$-subspace $D$ of $\Euc ^d$, consider the ``partial'' localization algebra $C^*_{G,L'}(D)$, that is, the closure of the $\ast$-algebra of elements $f \in C([0,\infty), C^*(D))$ such that there is $R \in C_0[0,\infty)$ such that $\xi f(t)\eta =0$ if the distance of $p (\supp \xi)$ and $p (\supp \eta)$ is larger than $R(t)$, where $p : \Euc ^d \to \Euc ^n$ is the projection.
In the same way as \citelist{\cite{MR1451759}\cite{MR2565716}}, we get 
\[\mu _{D,G,L'} : \lu \phi \K ^G_{*,c,\tau}(D) \to \lu \phi \K^Q_{*,c,\tau}(C^*_{G,L'}(D))\]
such that $\mu_{\Euc ^d,G}=(\ev_0)_* \circ \mu _{\Euc ^d, G,L'}$, where $\ev_0$ is the evaluation map. 

It suffices to show that $\mu _{\Euc ^d,G,L'}$ and $(ev_0)_*$ are isomorphisms.
Choose a simplicial decomposition $\tilde{K}$ of $\Euc ^{d-n}$ compatible with the $G$-action and let $K^{(k)}:=\tilde{K}^{(k)} \times \Euc ^{n}$.
By a Mayer-Vietoris argument as in the proof of Theorem 3.2 of \cite{MR1451759},$\mu _{K^{(k)},G,L'}$ are isomorphisms for any $k$. 
Note that the isomorphism for $\mu _{K^{(0)},G,L'}$ follows from the Real Baum--Connes isomorphism discussed above.
We can check that $(\ev_0)_*$ is an isomorphism in the same way as in Section 4 of \cite{MR1451759} by using the scaling map on $\Euc ^{d-n}$.
\end{proof}

\begin{lem}\label{lem:cliff}
Let $G$ be a finite group and let $V$ be a $d$-dimensional real representation of $G$. Then, there is $c'' \in \Hom (G,\Zt)$ and $\tau'' \in H^2(G,\Zt)$ such that the twisted equivariant $\K$-homology group $\lu\phi \K^G_{*,c,\tau}(V)$ is isomorphic to $\lu \phi \K_G^{\dim V-*,c',\tau'}(\pt)$ where $c':=c+c''$ and $\tau'=\tau + \tau''+\epsilon (c,c'')$. Moreover, $c''=0$ if and only if $\pi$ preserves the orientation of $V$ and $\tau''=0$ if and only if $G$ is a Pin representation (i.e.\ $G \to O(V)$ gives rise to $G \to \mathop{\mathrm{Pin}}(V)$). 
\end{lem}
\begin{proof}
By taking direct sum with trivial representations, we may assume that $\dim V=8k$. Then, the $\phi$-twisted $G$-action on $\Cl (V) \cong \Kop (\Delta _{V})$, where $\Delta _{V}$ is the unique $\Zt$-graded representation of $\Cl(V)$, gives rise to a $(\phi,c'',\tau'')$-twisted representation of $G$ on $\Delta _{V}$. It is actually induced from $G^{\tau''} \to \mathop{\mathrm{Pin}}(V)$ where the action of $g \in \mathop{\mathrm{Pin}}(V)$ on $\Delta _V$ is even if and only if $g \in \mathop{\mathrm{Spin}}(V)$ (see Section 1.2 of \cite{LawsonMichelsohn1989}).  

By Theorem 7 of \cite{MR582160}, we have the isomorphism $\lu \phi \KK^G_{c,\tau}(C_0(V),\real) \cong \lu \phi \KK^G_{c,\tau}(\real, \Cl (V))$. Now we obtain the desired isomorphism by (\ref{form:detw}).
\end{proof}

\begin{cor}\label{cor:cBC}
Assume that $\Pi$ is free abelian. Let $n:=\rank \Pi$, $\nu \in H^2_Q(B_\Pi, \mathbb{T}_\phi)$ as in Example \ref{exmp:ext}, $V:=T_{x_0}\Euc ^{d-n}$, $(c'',\tau'')$ the twist corresponding to $\Cl (V)$ and $c':=c+c''$, $\tau':=\tau +\tau'' +\epsilon (c,c'')$.
Then, there is an isomorphism
\[\Psi_{X,G} :\lu \phi \K _{*,c,\tau}^Q(C^*_G(X)^\Pi) \to \lu \phi \K ^{-*+n,c',\tau'+\nu}_Q(B_\Pi). \]
\end{cor}
In particular, for a product CT-type symmetry $(G \times \mathscr{A},\phi,c,\tau)$ such that $G$ is finite and acts on $\Euc ^d$ by a Spin representation,  $\lu \phi \K^{G \times \mathscr{A}}_{*,c,\tau}(C^*_G(X))$ is written in terms of the representation ring of $G$ as indicated in Example \ref{exmp:rep}.
\begin{proof}
By Theorem \ref{thm:cBC} and Lemma \ref{lem:cliff}, we obtain the isomorphism
\[ \lu \phi \K _{*,c,\tau}^Q(C^*(X)^\Pi ) \xleftarrow{\mu_{X,G}} \lu \phi \K _{*,c,\tau}^G (\Euc ^d) \cong \lu \phi \K _{*-n,c',\tau'}^G (\Euc ^{d-n}) \xrightarrow{\mu_{X',G}} \lu \phi \K _{*,c',\tau'}^Q(C^*(X')^\Pi ), \]
where $X'$ is a $G$-invariant Delone subset of $\Euc ^{d-n}$. It is the desired isomorphism by Example \ref{exmp:eqRoe}.
\end{proof}

\subsubsection{Coarse Mayer-Vietoris exact sequence}
Here we review the coarse analogue of the Mayer-Vietoris exact sequence introduced by Higson--Roe--Yu~\cite{MR1219916} (see also \cite{mathKT12120241}). 

\begin{thm}[Section 5 of \cite{MR1219916}]\label{thm:MV}
There is an exact sequence
\[
\begin{array}{l}
\cdots \to \lu \phi \K _{*,c,\tau}^{Q}(C^*_\Hilb (Z \subset X)^\Pi) \to \lu \phi \K _{*,c,\tau}^{Q}(C^*_\Hilb (Y_1 \subset X)^\Pi) \oplus \lu \phi \K _{*,c,\tau}^{Q}(C^*_\Hilb (Y_2 \subset X)^\Pi) \to \\
\hspace{1.5em} \to \lu \phi \K _{*,c,\tau}^{Q}(C^*_\Hilb (X)^\Pi) \xra{\partial_{\mathrm{MV}}}\lu \phi \K _{*-1,c,\tau}^{Q}(C^*_\Hilb (Z \subset X)^\Pi) \to \cdots.
\end{array}
\]
The boundary map $\partial _{\mathrm{MV}}: \lu \phi \K _{0,c,\tau}^{Q}(C^*_\Hilb (X)^\Pi) \to \lu \phi \K _{-1,c,\tau}^{Q} (C^*_\Hilb (Z)^\Pi)$ is given by
\[ \partial _{\mathrm{MV}}[s]=[\sigma (\hat{s})] \in \lu \phi \K _{-1,c,\tau}^{Q} (C^*_\Hilb (Z \subset X)^\Pi) \]
where $\hat{s}:=\pi (\chi _{Y_1})s \pi (\chi _{Y_1}) \in C^*_\Hilb (Y_1)$.
\end{thm}
Note that the inclusion $C^*_{\Hilb}(A)^\Pi \to C^*_{\Hilb}(A \subset X)^\Pi$ induces the isomorphism of twisted equivariant $\K$-groups (cf.\ Lemma 1 of \cite{MR1219916}).
\begin{proof}
Set $A:=C^*_\Hilb (X)^\Pi$, $I_k:=C^*_\Hilb (Y_k \subset X)^\Pi$ for $k=1,2$ and $J:=I_1 \cap I_2$. Then, by assumption we have $J=C^*_\Hilb (Z \subset X)^\Pi$ and $I_1+I_2=A$. Note that $A/J\cong  I_1 / J \oplus I_2 / J$. We write $q_k: A_k \to A_k/J$ and $\theta: A \to A/J$ for the quotient and $p _k$ for the projection $A/J \to I_k/J$. 

Now, the coarse Mayer-Vietoris exact sequence is obtained as the long exact sequence associated to the short exact sequence 
\[0 \to SA \to \Omega(A;I_1,I_2)  \to I_1 \oplus I_2 \to 0 \]
where
\[
\Omega (A;I_1,I_2) := \{  f\in C([0,1],A) \mid f(0) \in I_1 , \ f(1) \in I_2 \}. 
\]
Actually, since $\Omega(A;I_1,I_2) /C([0,1],J) \cong C_0([0,1),I/J_1) \oplus C_0((0,1],I/J_2)$ is contractible, we have $\lu \phi \K_{*,c,\tau}^{Q}(\Omega (A;I_1,I_2) ) \cong \lu \phi \K_{*,c,\tau}^Q(J)$. 

Let $\psi$ be the $\ast$-homomorphism from $\Omega(A;I_1,I_2)$ to
\[
\cone (q_1) := \{ (a,f) \in I_1 \oplus C([0,1), I_1/J) \mid f(0)=q_1(a) \}
\] 
given by $\psi (f)= (f(0), p_1 \circ \theta (f))$. Then, the diagram
\[
\xymatrix{
SA \ar[r]^{b} \ar[d]^{p_1 \circ \theta} &\Omega (A;I_1,I_2) \ar[d]^\psi   & J \ar[l]^i \ar[ld]^{i'} \\
S(I_1/J) \ar[r]^{b'} & \cone (q_1) &
}
\]
commutes. Now, by definition
\[
\partial _{\mathrm{MV}}=(i_*)^{-1}\circ b_*: \lu \phi \K _{*,c,\tau}^{Q} (SA) \to \lu \phi \K _{*,c,\tau}^{Q}(\Omega (A:I_1,I_2)) \cong \lu \phi \K _{*,c,\tau}^{Q}(J),\]
and hence we obtain
\[ \partial _{\mathrm{MV}}=(i_*')^{-1}\circ b'_* \circ (p_1 \circ \theta)_* =\partial \circ (p _1 \circ \theta)_*,\]
where $\partial$ is the boundary map of the exact sequence $0 \to J \to I_1 \to I_1/J \to 0$. 
\end{proof}

\begin{cor}\label{cor:Toep}
Let $\partial _{\mathrm{Toep}}$ be the boundary map of the Toeplitz extension
\[
0 \to A \otimes \Kop \to \mathcal{T}(A) \to \zahl \ltimes A \to 0
\]
where $A:=\zahl^{d-1}\ltimes c_b(\zahl^d)$. Then, we have $\partial _{\mathrm{MV}}=f_* \circ \partial _{\mathrm{Toep}}$, where $f: \zahl ^{d-1} \ltimes c_b(\zahl ^d) \to \zahl ^{d-1} \ltimes c_b(\zahl ^{d-1})$ be the $\ast$-homomorphism of crossed products induced from the restriction. 
\end{cor}
In particular, $\partial _{\mathrm{MV}} \circ \mathop{\mathrm{Res}} =\mathop{\mathrm{Res}} \circ \partial_{\mathrm{Toep}}$ on $\lu \phi \K_{0,c,\tau}^G(C^*_r(\zahl ^d))$, where $\partial_{\mathrm{Toep}}$ is the boundary map of the Toeplitz exact sequence
\[0 \to C^*_r(\zahl ^{d-1}) \otimes \Kop \to C^*_r(\zahl ^{d}) \otimes \mathcal{T} \to C^*_r(\zahl ^d) \to 0\]
and $\mathrm{Res}$ is the map induced from the inclusion $C^*_u(|\zahl ^d|)^{\zahl ^d}\to C^*_u(|\zahl ^d|)$.

\begin{proof}
We identify $\zahl \ltimes A$ with the subalgebra of $\Bop (\ell ^2 (\zahl ^d)) \otimes \Bop(\ell ^2 (\zahl))$ generated by $A \otimes 1$ and $u \otimes v$ where $u$ is a generator of $C^*_r\zahl \subset \zahl \ltimes A$ and $v$ is the birateral shift operator on $\ell ^2 (\zahl)$. Then, $\mathcal{T}(A)$ is the $\Cst$-algebra generated by $(1 \otimes P)(\zahl \ltimes A)(1 \otimes P)$ where $P$ is the projection onto $\ell ^2(\zahl_{\geq 0})$. Let $p$ be the projection onto $\ell ^2 (\zahl ^{d-1} \times D_{\zahl } )$ where $D _\zahl$ is the diagonal subset of $\zahl \times \zahl$. Then, $p$ commutes with $\mathcal{T}(A)$ and the exact sequence 
\[
0 \to p(A \otimes \Kop)p \to p(\mathcal{T}(A))p \to p(\zahl \ltimes A)p \to 0
\]
is the same thing as $0 \to C^*_u(Z \subset Y_1) \to
 C^*_u(Y_1) \to C^*_u(Z,Y_1) \to 0$ and the projection $A \otimes \Kop \to p(A \otimes \Kop)p$ is the same thing as $f$.
\end{proof}

\subsubsection{$\K$-cycles and cyclic cocycles}
At the end of this section, we study fundamental cycles of the coarse $\Cst$-algebras for $G$-invariant Delone subsets of $\Euc ^d$. For simplicity, let $G$ be a finite group.

Let $V:=T_x\Euc ^d$ as in Example \ref{exmp:ext}, $V^-:=V$ with the negative definite inner product $-\ebk{\blank,\blank}$ and let $\tilde{V}:=V \oplus V^-$. 
Then, $\Cl(\tilde{V})$ has a unique irreducible $\Zt$-graded representation $\Delta _{\tilde{V}}$. 
For a finite $(X,G)$-module $(\Hilb , U, \pi)$, we obtain a twisted equivariant $\K$-cycle
\[
\Xi _{X,\Hilb}:=[\Hilb \hotimes \Delta _{\tilde{V}},\id , F] \in \lu \phi \KK^G (C^*_\Hilb (X)^\Pi \hotimes \Cl(V^-) , \real ).
\]
where $\id$ means the canonical inclusion $C^*_\Hilb (X)^\Pi \hotimes \Cl(V^-) \subset \Bop (\Hilb) \hotimes \Bop (\Delta _{\tilde{V}})$ and 
\[
F:=D(1+D^2)^{-1/2}, D:= \sum \nolimits_{j=1}^{d} \pi (x_j) \hotimes e_j \in \Bop (\Hilb ),
\]
where $x_j$ are coordinate functions on $X \subset \Euc ^d$ and $e_j$ are Clifford generators of $\Cl(V)$. 
Note that $C^*_\Hilb(X)$ are unital and in particular $\sigma$-unital.
Therefore, as in Remark \ref{remk:KK}, the twisted Kasparov product induces the group homomorphism
\[
\Xi _{X,\Hilb}^*:=\blank \otimes \Xi _{X,\Hilb}: \lu \phi \K ^G_{0,c,\tau}(C^*_\Hilb (X)) \to \lu \phi \KK^G_{0,c,\tau}(\Cl(V^-),\real) \cong  \lu \phi \K^G_{0,c,\tau}(\Cl (V)).
\]
Here, $\lu \phi \K ^G_{0,c,\tau}(\Cl(V))$ is isomorphic to $\lu \phi\K _G^{0,c',\tau'}(\pt)$ by Lemma \ref{lem:cliff}.
For example, for a CT-symmetry, $\Xi _X^*$ takes value in the group as in Table \ref{table:Kitaev}. 
In particular, we simply write $\Xi _X:=\Xi _{X,\Hilb _{X,G}} $.
\begin{prp}\label{prp:index}
The diagram
\[
\xymatrix@R=1.2em{
\lu \phi \K ^G_{0,c,\tau} (C^*_{u,G}(X)) \ar[d] \ar[r]^{\Xi _X^*} & \lu \phi \K_G^{d,c',\tau'}(\pt) \\
\lu \phi \K ^G_{0,c,\tau}(C^*_G(X))  \ar[ru] _{\Psi_{X,G}} &
}
\]
commutes, where ${\Psi _{X,G}}$ is the isomorphism given in Corollary \ref{cor:cBC}.
\end{prp}

\begin{proof}
Since $C^*_G(X)=\overline{\bigcup _{\Hilb } C^*_\Hilb (X)}$ where $\Hilb$ runs over all finite $(X,G)$-submodules of $\Hilb_{X,G}^\infty$, we have the isomorphism
\[
\varinjlim \lu \phi \K_{0,c,\tau}^{Q}(C^*_\Hilb (X) ) \to \lu \phi \K_{0,c,\tau}^{Q}(C^*(X) )
\]
in the same as continuity of $\K$-groups. Since $\Xi _{X,\Hilb}^*$ commutes with the inclusion $C^*_{\Hilb}(X) \subset C^*_{\Hilb'}(X)$ for $\Hilb \subset \Hilb'$, $\Xi _X^*$ factors through $\lu \phi \K ^Q_{0,c,\tau}(C^*_G(X))$.

By Corollary \ref{cor:cBC}, any element in $\lu \phi \K^G_{0,c,\tau}(C^*_G(X))$ is of the form $\mu(\alpha) \otimes \xi$ for some $\xi \in \lu \phi \K_{-d,c,\tau}^G(\real)$, where $\alpha \in \lu \phi \K _{d}^Q(\Euc ^d)$ is the element represented by the de Rham operator on $\Euc ^d$ (see Section 5.II of \cite{MR582160}). 
Therefore, it suffices to show that $\Xi _X^*(\mu (\alpha))=1 \in \lu \phi \KK^Q(\real, \real)$. 
Actually, it is proved in the same way as Theorem 1 of \cite{MR1435703} that $\Xi _X^*(\mu (\alpha))$ coincides with $\beta \otimes _{C_0(\Euc ^d) \otimes \Cl_{0,d}} \alpha =1$, where $\beta$ is the Bott element. 
\end{proof}

\begin{remk}\label{remk:genind}
For a general $(G,\phi,c,\tau)$ satisfying Assumption \ref{assump:ext}, we obtain a similar index pairing. Recall that we have a decomposition $\Euc ^d \cong \Euc ^{n-d} \times \Euc ^{n}$ as in Example \ref{exmp:extRoe}.
By Lemma \ref{lem:Morita}, we may assume that $X=X' \times X''$ where $X' \subset \Euc ^{d-n}$ and $X'':=G \cdot x'' \subset \Euc ^n$ are $G$-invariant Delone subsets. 
Then, we have an isomorphism
\[C^*_{G}(X)^\Pi \cong C^*_{Q}(X') \otimes C^*_{G}(X'')^\Pi  \cong C^*_{Q}(X') \otimes \Kop (\ell^2_\Pi(G)).\]
Therefore, by using $\Xi _{X'}$, we obtain an index pairing which takes value in $\lu \phi \K ^Q_{-d,c',\tau'}(\Kop (\ell^2_\Pi(G))) \cong \lu \phi \K^{d-n,c',\tau' +\nu}_Q(B_\Pi)$.
By the same proof as Proposition \ref{prp:index}, we obtain $\Xi _{X'}^*=\Psi_{X,G}$ for this case.
\end{remk}

\begin{cor}\label{cor:iota}
Let $G$ be a discrete group as in Assumption \ref{assump:ext} such that $\Pi$ is free abelian ($n:=\rank \Pi$) and let $X$ be a $G$-invariant Delone subset of $\Euc ^d$. 
Let $\iota: c_b(X; \Kop (\Hilb _{X,G}^\infty)) \to C^*_G(X)^\Pi$ be the inclusion. Under the isomorphism in Corollary \ref{cor:cBC}, 
\begin{itemize}
\item $\Im \iota_*=0$ when $n<d$.
\item $\Im \iota_*$ is spanned by $\lInd _{G_x}^Q  (\lu \phi \K _{G_x}^{*,c,\tau}(\pt ))$ when $n=d$. In particular, $\Im \iota _* \subset \mathrm{Triv}$.
\end{itemize}
\end{cor}
\begin{proof}
When $n<d$, it follows from $\Xi _X^* (\Im \iota_*)=0$. When $n=d$, it follows from the definition of the $\ell^2$-induction. 
\end{proof}

Next, we relate it with the coarse Mayer--Vietoris exact sequence. For this, let $G$ be a finite group acting on $\Euc ^{d-1}$. 
We choose a decomposition $\Euc ^1=\Euc ^1_+ \cup \Euc ^1_-$ such that $\Euc ^1_+ \cap \Euc ^1_- =\pt$ and set $\Euc ^d_\pm:= \Euc ^{d-1} \times \Euc ^1_{\pm}$.
\begin{cor}\label{cor:MVind}
Let $X=Y_+ \cup Y_-$ be a Delone partition of $\Euc ^d=\Euc ^d_+ \cup \Euc ^d_-$. Set $Z:=Y_+ \cap Y_-$. Then, we have
\[
\Xi _X^*[s] = \Xi _Z^* \partial _{\mathrm{MV}}[s] \in \lu \phi \K_G^{d,c',\tau'}(\pt).
\]
\end{cor}
In the same way as Remark \ref{remk:genind}, we obtain the same result when $G$ is a general discrete group.
\begin{proof}
By Proposition \ref{prp:index}, it suffices to show that $\Psi_{Z,G} \circ \partial _{\mathrm{MV}}=\Psi _{X,G}$.
It follows from the fact that the coarse Baum-Connes isomorphism is compatible with the Mayer-Vietoris exact sequences (see for example Section 3 of \cite{mathKT12120241}).
\end{proof}

For the case of type {\typeA} or {\typeAIII} CT-symmetry, this pairing is also related to the pairing with the fundamental cyclic cocycle.
Let $\omega \in \beta \nat \setminus \nat$ be a free ultrafilter and let $\mathbf{P} _\omega \in c_b(X)^*$ be the corresponding probability measure on $\beta X$, that is, 
\[
\int f(x)d\mathbf{P} _\omega (x)  =\lim _{n \to \omega } |\Lambda _n |^{-1} \sum\nolimits _{x \in \Lambda _n} f(x)
\] 
where $\Lambda _n:=[-n,n] \cap X$.
As in the proof of Lemma 4.7 of \cite{MR2007488}, the state 
\[\mathcal{T}^\omega (a):=\int d\mathbf{P} _\omega (x) (\Tr (a_{xx}))  \] 
on $C^*_uX$ is tracial since $\Lambda _n$ is a F{\o}lner set of an amenable uniformly discrete metric space $X$. 
Hence, we obtain a cyclic $d$-cocycle
\maa{\mathcal{T} _d^\omega (a_0,a_1,\cdots , a_d):=\sum\nolimits_{\sigma \in\mathfrak{S}_n} (-1)^{|\sigma|}\mathcal{T} ^\omega(a _0 \nabla _{\sigma (1)}(a_1) \cdots \nabla _{\sigma (d)}(a_d)) \label{form:cyc}}
on the dense $\ast$-subalgebra $\mathcal{B}_u(X):=\{ a \in C^*_u(X) \mid \nabla_j (a) \in C^*_u(X) \}$ of $C^*_u(X)$, where $\nabla _j(a):=i[\pi(x_j),a]$. Since $\mathcal{B}_u(X)$ has the same $\K$-theory as $C^*_u(X)$ (see Appendix 3 of \cite{MR823176}), this pairing induces the group homomorphism $\ebk{\blank, \mathcal{T}^\omega_d}:\K _d(C^*_uX) \to \real$. 

\begin{lem}\label{lem:pair}
We have $\langle \xi , \mathcal{T}_d^\omega \rangle =\Xi_X^* \xi$ for any $\xi \in \K_d(C^*_u(X))$. 
In particular, $\langle \blank, \mathcal{T}_d^\omega \rangle $ is independent of the choice of $\omega \in \beta \nat \setminus \nat$.
\end{lem}
\begin{proof}
It suffices to show that the Chern-Connes character of $\Xi_X$ coincides with $\mathcal{T}^\omega_d$ as cyclic cocycles. The proof is given in the same way as Theorem 10 of \cite{MR1295473}.
\end{proof}

Let $X=Y_+ \cup Y_-$ be a Delone partition of $\Euc ^d =\Euc ^d_+ \cup \Euc _-^d$.
In the same way as above, we obtain a cyclic $(d-1)$-cocycle 
\[
\hat{\mathcal{T}} _{d-1}^\omega (a_0,a_1,\cdots , a_{d-1}):=\sum_{\sigma \in \mathfrak{S}_{d-1}} (-1)^{|\sigma|}\mathcal{T}^\omega  (a _0 \nabla _{\sigma (1)}(a_1) \cdots \nabla _{\sigma (d-1)}(a_{d-1}))
\]
on a dense subalgebra of $C^*_u(Z \subset Y_1)$ where $Z:=Y_1 \cap Y_2$. 
By Corollary \ref{cor:MVind} and Lemma \ref{lem:pair}, we have 
\maa{\langle \xi, \mathcal{T}^\omega_d\rangle =\langle \partial_{\mathrm{MV}}(\xi), \hat{\mathcal{T}}^\omega_{d-1}\rangle. \label{form:MVcyc}}

\section{Controlled topological phases}\label{section:3}
In this section, we introduce the notion of controlled topological phase for bulk and edge quantum systems with an arbitrary symmetry of quantum mechanics in the sense of \cite{MR3119923}. 
Let $(G,\phi,c,\tau)$ be as Assumption \ref{assump:ext} and let $X$ be a proper metric space with proper isometric $G$-action.
\begin{defn}\label{defn:bulkTP}
We say that a \textit{bulk quantum system controlled at $X$} with the symmetry class $(G,\phi , c,\tau )$ is a quadruple $(\Hilb, U, \pi , H)$ where 
\begin{itemize}
\item $(\Hilb,U,\pi)$ is a uniformly finite $(\phi,c,\tau)$-twisted $(X,G)$-module,
\item $H \in \Bop (\Hilb)$ is a $c$-twisted $G$-invariant controlled self-adjoint operator with a spectral gap at $[-\varepsilon, \varepsilon]$ for some $\varepsilon >0$.
\end{itemize}
A \emph{controlled bulk topological phase} is a equivalence class of bulk quantum systems controlled at $X$ with the symmetry class $(G,\phi,c,\tau)$ with respect to the equivalence relation generated by the following:
\begin{itemize}
\item[(1)] $(\Hilb, U, \pi ,H) \sim (\Hilb ', U', \pi' ,H')$ if there is an isomorphism $V: \Hilb  \to \Hilb '$ of $(\phi,c,\tau)$-twisted $(X,G)$-modules such that $VHV^*=H'$,
\item[(2)] $(\Hilb, U, \pi ,H_0) \sim (\Hilb, U, \pi ,H_1)$ if there is a continuous path $H_t \in \Bop (\Hilb)$ of $c$-twisted $G$-invariant controlled self-adjoint operators with a spectral gap at $[-\varepsilon, \varepsilon]$ for some $\varepsilon >0$,
\item[(3)] $(\Hilb, U, \pi ,H) \sim (\Hilb \oplus \Hilb ', U \oplus U', \pi \oplus \pi ' ,H \oplus H')$ if $H'$ commutes with $\pi (c_0(X))$.
\end{itemize}
We write $\TP (X;G,\phi , c,\tau )$ for the set of bulk topological phases. For a CT-type $(\mathscr{A},\tau)$ with the Cartan label $\mathrm{L}$, we simply write $\TP(X;\mathrm{L}):=\TP(X;\mathscr{A},\phi,c,\tau)$ and $\TP(X;G,\mathrm{L}):=\TP(X;G \times \mathscr{A},\phi,c,\tau)$.
\end{defn}

\begin{remk}
In this framework, we can deal with a Hamiltonian $H$ with a long range term such that $\ssbk{H_{xy}}$ decays as $d(x,y) \to \infty$ uniformly. Actually, such $H$ is approximated by the controlled operators $H_R$ given by 
\[ (H_R) _{xy}= \begin{cases}H_{xy} & d(x,y)<R, \\  0 & d(x,y)\geq R, \end{cases}\]
in the norm topology.
For a sufficiently large $R>0$, $\ssbk{H-H_R}$ is small enough that $H_R$ has a spectral gap at $0$ and all $H_{R'}$ with $R'>R$ are in the same topological phase.
\end{remk}

\begin{remk}
Here, we assume that the Hilbert space $\Hilb$ of quantum states is equipped with an a priori $\Zt$-grading. In the case of {\typeAIII} systems, it corresponds to an a priori choice of reference isomorphisms $h_{\mathrm{ref}}$ in \cite{mathph150404863}. 
This assumption is not artificial because $\Hilb$ has a canonical $\Zt$-grading for a large class of examples. 
For example, when the system has the particle-hole symmetry such as Bogoliubov--de-Gennes Hamiltonians, $\Hilb$ is decomposed as the direct sum of the particle and hole components. 

When $\Hilb$ does not have a canonical $\Zt$-grading, we only have to fix an arbitrary choice of such $\Zt$-gradings. 
For example, when we consider type A, AI or {\typeAII} symmetry class, we can use the trivial grading $\gamma_{\Hilb} =1$. 
In fact, we give the relation (3) in order for the topological phase of a system to be independent of the choice of $\Zt$-gradings of $\Hilb$ which is compatible with $U$ and $\pi$. 
Actually, let $\Hilb  '$ be the Hilbert space $\Hilb$ with another $\Zt$-grading $\gamma$ compatible with $U$ and $\pi$. Then, 
\ma{
(\Hilb, U,\pi,H) &\sim (\Hilb, U, \pi ,H) \oplus (\Hilb', U, \pi,\gamma) \sim (\Hilb, U,\pi,\gamma ) \oplus (\Hilb ', U,\pi,H)\\
& \sim (\Hilb ', U,\pi,H),
}
where the second equivalence is given by the homotopy $R_t(H \oplus \gamma)R_t^*$ where $R_t$ are the $2$-dimensional rotation matrices. 

On the other hand, a general quantum system with sublattice symmetry does not admit any compatible $\Zt$-grading. 
For example, let $X_1$ and $X_2$ be mutually disjoint Delone subsets of $\Euc ^d$ and let $X:=X_1 \sqcup X_2$. 
If $P$ acts on each $\Hilb _x$ as $+1$ for $x \in X_1$ and $-1$ for $x \in X_2$, there is no compatible $\Zt$-grading on $\Hilb$.
However, even in this case, sometimes we can get a controlled bulk quantum system by deforming $X$.
More precisely, if there is another $G$-invariant Delone subset $X'$ of $\Euc ^d$ and a $G$-map $f : X \to X'$ such that $d(x,f(x))$ is uniormly bounded and $\sum_{x \in f^{-1}(x')} \dim \Hilb _{x}^0=\sum_{x \in f^{-1}(x')} \dim \Hilb _{x}^1$ for any $x' \in X'$, we can define a bulk quantum system $(\Hilb, U, \pi \circ f^*, H)$ controlled at $X'$. 
For example, in the case of a honeycomb lattice,  we get a map $X \to X_1$ by shifting $X_2$ as in Section 3.1 of \cite{MR3123539}.
\end{remk}

\begin{prp}\label{prp:bulk}
We have the isomorphism 
\[
\TP(X;G,\phi,\tau , c) \cong \lu \phi \K _{0,\tau,c}^Q (C^*_{u,G}(X)^\Pi )/\Im (\iota_u) _*,
\]
where $\iota_u : c_b(X,\Kop (\Hilb _{X,G}))^\Pi \to C^*_{u,G}(X)^\Pi$ is the inclusion.
\end{prp}
In particular, it is independent of the choice of full $G$-invariant Delone subsets $X$ of $M$ up to $\Im (\iota _u)_*$ by Lemma \ref{lem:Morita}.

\begin{proof}
We have a well-defined map
\[ \Theta : \lu \phi \K _{0,\tau,c}^Q (C^*_{u,G}(X)^\Pi ) \to \TP(X;G,\phi,c,\tau), [s] \mapsto (\Hilb _{X,G}\otimes \mathscr{V} , U,\pi , s ). \]

It is suejective because any $\Hilb $ is contained in some $\Hilb _{X,G} \otimes \mathscr{V}$ and 
\[(\Hilb , U,\pi , H) \sim (\Hilb _{X,G}\otimes \mathscr{V} , U,\pi , H \oplus \gamma _{\Hilb ^\perp} )=\Theta [H \oplus \gamma _{\Hilb ^\perp}].\] 
By the relation (3), $\Im \iota_* \subset \Ker \Theta $.
To see that $\Ker \Theta = \Im \iota _*$, it suffices to check $[usu^*]=[s]$ for any even $G$-invariant unitary $u \in c_b(X;\Kop (\Hilb _{X,G} \otimes \mathscr{V}))$. Set 
\[ v_t:=(u \oplus 1)R_t(u^* \oplus 1) R_t^* \in c_b(X, \Kop ((\Hilb_{X,G} \otimes \mathscr{V})^2)). \]
Then, since $v_0 =1 \oplus 1$ and $v_1=u \oplus u^*$, we get a path $v_t(s \oplus \gamma)v_t^*$ connecting $s \oplus \gamma$ and $usu^* \oplus \gamma $, which shows $[s]=[usu^*]$.
\end{proof}

\begin{remk}\label{rmk:twist}
Even if Assumption \ref{assump:ext} is not satisfied such as the case that the translation symmetry is twisted by the magetic field, we can consider cotroled topological phases in the same way. 
In this case, controlled topological phases are classified by the twisted equivariant $\K$-group of ``twisted'' invariant uniform Roe algebras defined by using twisted $(X,G)$-modules.  
For example, when $G=\zahl ^d$, $X=|\zahl ^d|$ and $\tau$ is nontrivial, this $\Cst$-algebra is the same thing as the noncommutative torus $C^*_{-\tau, r}(\zahl ^d)$ used in the previous researches~\citelist{\cite{MR1303779}\cite{MR1295473}\cite{MR1877916}}. 
\end{remk}

\begin{exmp}
When $\Pi=\zahl ^d$ and $X=|\zahl ^d|$, then $\TP(X;\Pi,{\mathrm{L}})$ is isomrophic to  the reduced $\KL ^0$-group of the Brillouin torus $B_\Pi$ because $C^*_u(X)^\Pi \cong C^*_r(\Pi)$ as is described in Example \ref{exmp:eqRoe}.  

More generally, when $G$  is a cocompact discrete Euclidean group, $\TP(X;G,{\mathrm{L}})$ is isomorphic to $\KL^{0,\nu}_Q (B_\Pi)/\mathrm{Triv}$, which is essentially the same thing as the set of reduced topological phases $\mathscr{RTP}_F(G,\phi,\tau,c)$ of type $F$ in \cite{MR3119923} (note that they deal with Galilean spacetimes). 
\end{exmp}

\begin{exmp}
For a Delone subset of $\Euc ^d$, the group $\TP(X;{\mathrm{L}})$ is isomorphic to $\KL_0(C^*_u(X))/\Im (\iota_u)_*$. By Lemma \ref{lem:Morita}, they are isomorphic to the groups calculated in Lemma \ref{lem:uRoe} when $\mathrm{L}$ is a real Cartan label.
\end{exmp}

\begin{remk}\label{remk:compare}
The set $\TP(X;L)$ is related to the classification of topological phases via $\K$-groups of noncommutative Brillouin torus.
Recall that the noncommutative Brillouin torus $A:=\zahl ^d \ltimes C(\Omega)$ introduced in \cite{MR1295473} is a subalgebra of $C^*_u(X)$ (Example \ref{exmp:groupoid} and Lemma \ref{lem:groupoid}). 
For any gapped Hamiltonian $H \in A\hotimes \Kop(\mathscr{V})$ as in \citelist{\cite{MR1295473}\cite{mathph14067366}\cite{mathKT150906271}}, the inclusion $j: A \to C^*_u(X)$ induces a homomorphism
\[\KL_0(A) \to \KL_0 (C^*_{u}(X) ) \to \TP(X;L)\]
mapping $[H]$ to $[\ell^2(X)\hotimes \mathscr{V},U,\pi,H]$. 
We also remark that $j_* \Xi _X \in \KR^d(A)$ is the same thing as the $\KR$-cycle used in \cite{GrossmannSchulz-Baldes} and \cite{mathKT150907210}.
\end{remk}

\begin{exmp}
Let $G$ be a discrete group acting on the hyperbolic space $M=SO^+(1,d)/SO(d)$ properly and isometrically and let $X$ be a $G$-invariant Delone subset of $M$. Then, the Hamiltonian studied in \citelist{\cite{MR1600480}\cite{MR1684000}} determines an element in $\TP (X;G, \mathrm{A})$ when the $2$-cocycle is trivial.
\end{exmp}

Next we introduce the notion of edge topological phase. In this paper, we use the hat symbol for objects in edge quantum systems instead of objects related to $\Zt$-gradings. Let $Z \subset Y$ be a pair of metric spaces with a proper isometric $G$-action.  Here we write $q: C^*_\Hilb(Y) \to C^*_\Hilb(Y,Z)$ for the quotient.
\begin{defn}\label{defn:edgeTP}
We say that an \textit{edge quantum system controlled at $Z \subset Y$} with the symmetry class $(G,\phi , c, \tau )$ is a quadruple $(\hat{\Hilb}, \hat{U}, \hat{\pi} ,\hat{H})$ where 
\begin{itemize}
\item $(\hat{\Hilb}, \hat{U}, \hat{\pi})$ is a uniformly finite $(\phi,c,\tau)$-twisted $(Y,G)$-module,
\item $\hat{H} \in \Bop (\hat{\Hilb})$ is a $c$-twisted $G$-invariant controlled self-adjoint operator such that $q (\hat{H}) \in C^*_{\Hilb}(Y,Z)$ has a spectral gap at $[-\varepsilon,\varepsilon]$ for some $\varepsilon >0$.
\end{itemize}
A \emph{controlled edge topological phase} is a equivalence class of edge quantum systems controlled at $Z \subset Y$ with the symmetry class $(G,\phi,c,\tau)$ with respect to the equivalence relation generated by the following:
\begin{itemize}
\item[(1)] $(\hat{\Hilb}, \hat{U}, \hat{\pi} ,\hat{H}) \sim (\hat{\Hilb} ',\hat{ U}', \hat{\pi}' ,\hat{H}')$ if there is a unitary $V: \hat{\Hilb}  \to \hat{\Hilb} '$ such that $V\hat{U}_gV^* =\hat{U}_g'$, $V\hat{\pi} (f)V^*=\hat{\pi} '(f)$ and $V\hat{H}V^*=\hat{H}'$,
\item[(2)] $(\hat{\Hilb}, \hat{U}, \hat{\pi} ,\hat{H}_0) \sim (\hat{\Hilb}, \hat{U}, \hat{\pi} ,\hat{H}_1)$ if there is a continuous path $\hat{H}_t \in \Bop (\hat{\Hilb})$ of $c$-twisted $G$-invariant controlled self-adjoint operator such that $q (\hat{H}) \in C^*_u(Y,Z)$ has a spectral gap at $[-\varepsilon,\varepsilon]$ for some $\varepsilon >0$,
\item[(3)] $(\hat{\Hilb}, \hat{U}, \hat{\pi},\hat{H}) \sim (\hat{\Hilb} \oplus \hat{\Hilb} ', \hat{U} \oplus \hat{U}', \hat{\pi} \oplus \hat{\pi} ' ,\hat{H} \oplus \hat{H}')$ if $\hat{H}'$ is invertible.
\end{itemize}
We write $\hat{\TP} (Z \subset Y;G,\varphi , \tau ,c)$ for the set of controlled edge topological phases.
\end{defn}

\begin{prp}\label{prp:edge}
We have the isomorphism
$$\hat{\mathscr{TP}}(Z \subset Y;G,\phi,\tau , c) \cong \Im \partial \subset  \lu \phi \K _{-1, \tau,c}^{Q} (C^*_{u,G}(Z \subset Y)^\Pi)$$
where $\partial $ is the boundary map of the long exact sequence associated to
\[0 \to C^*_{u,G}(Z \subset Y)^\Pi \to C^*_{u,G}(Y)^\Pi \to C^*_{u,G}(Y,Z)^\Pi \to 0.\]
\end{prp}
In particular, it is independent of the choice of $G$-invariant Delone pairs $Z \subset Y$ of $N \subset M$ by Lemma \ref{lem:Morita}.
\begin{proof}
We have a well-defined map
\[ \hat{\Theta } : \lu \phi \K _{-1, \tau,c}^{Q} (C^*_{u,G}(Y,Z)^\Pi) \to \hat{\TP}(Z \subset Y;G,\phi,\tau , c), [s] \mapsto (\Hilb_{Y,G}\otimes \mathscr{V}, U,\pi, \hat{s}) \]
where $\hat{s}$ is a self-adjoint $c$-twisted $Q$-invariant lift of $s$ in $C^*_{u,G}(Y)^\Pi$.
In the same way as Proposition \ref{prp:bulk}, we obtain that $\Ker \hat \Theta =\Im q_*$ and $\hat \Theta$ is surjective.
Finally, we obtain the conclusion since $\lu \phi \K _{0,\tau,c}^Q(C^*_{u,G}(Y,Z)^\Pi ) / \Im q_* \cong \Im \partial$. 
\end{proof}

\section{Bulk-edge correspondence}\label{section:4}
In this section, we give the definition of the bulk and edge indices of topological phases controlled at Delone subsets of the Euclidean space $\Euc ^d$ and formulate the bulk-edge correspondence for them. 

Let $(G,\phi,c,\tau)$ be as in Assumption \ref{assump:ext} such that $\Pi$ is free abelian (particularly we take notice of the case as in Example \ref{exmp:extRoe}) and let $X$ be a $G$-invariant Delone subset of $\Euc ^d$. 
\begin{defn}\label{defn:bulk}
We write $\bind$ for the group homomorphism
\[
\bind : \mathscr{TP}(X;G,\phi,c, \tau) \to  \lu \phi \K_Q^{n-d,c',\tau'+\nu }(B_\Pi)/\mathrm{Triv}_{d-n} 
\]
induced from the inclusion $C^* _{u,G}(X)^\Pi \to C^*_G(X)^{\Pi}$ and the isomorphisms in Proposition \ref{prp:bulk}, Corollary \ref{cor:cBC} and Corollary \ref{cor:iota}. Here, we write $\mathrm{Triv}_0:=\mathrm{Triv}$ and $\mathrm{Triv}_k:=0$ for $k \neq 0$.
\end{defn}
When $n=d$, it is merely the quotient map 
\[\lu \phi \K_Q^{n-d,c',\tau'+\nu }(B_\Pi)/\Im (\iota_u)_* \to \lu \phi \K_Q^{n-d,c',\tau'+\nu }(B_\Pi)/\mathrm{Triv} \]
since $C^*_G(X)^\Pi \cong \Kop(\ell^2_\Pi(G)) \otimes \Kop$ by Example \ref{exmp:extRoe}. 
On the other hand, when $n=0$, $\bind=\Xi _X^*$ by Proposition \ref{prp:index}. 
In particular, for CT-symmetries, $\bind$ is the same thing as the index pairing introduced in \cite{GrossmannSchulz-Baldes} by Remark \ref{remk:CTvsR}.

Next, we give a definition of the edge index. 
Let $\Euc ^d=\Euc ^d_+ \cup \Euc ^d_-$ be as in Corollary \ref{cor:MVind} and 
let $Z \subset Y$ be a $G$-invariant Delone pair (Definition \ref{defn:Delone}) in $\Euc ^{d-1} \subset  \Euc ^{d}_+$.

\begin{defn}\label{defn:edge}
We write $\eind$ for the group homomorphism
\[
\eind : \hat{\mathscr{TP}}(Z \subset Y;G,\phi,c', \tau') \to \lu \phi \K_Q^{n-d,c',\tau'+\nu }(B_\Pi)
\]
induced from the inclusion $C^* _{u,G}(Z \subset Y)^{\Pi} \to C^*_G(Z \subset Y)^{\Pi}$, Proposition \ref{prp:edge} and Corollary \ref{cor:cBC}. 
\end{defn}

Finally, we show the bulk-edge correspondence for these indices. Let $X=Y_+ \cup Y_-$ be a Delone partition of $\Euc ^d=\Euc ^d_+ \cup \Euc ^d_-$ and $Z:=Y_1 \cap Y_2$.
Since $\rank \Pi <d$, the bulk and edge indices take values in the same group $\lu \phi \K^{n-d,c',\tau'}_Q(B_\Pi)$.
For a bulk quantum system $(\Hilb, U, \pi , H)$ controlled at $X$ with the symmetry class $(G,\phi,c,\tau)$, we obtain the edge system of the same material as
\[
\hat{\Hilb}:=\pi (\chi _{Y_+}) \Hilb, \ \hat{U}_g:= U_g |_{\hat{\Hilb}} , \ \hat{\pi} (f):=\pi (f)|_{\hat{\Hilb}}, \hat{H}:=\pi(\chi _{Y_+}) H \pi(\chi _{Y_+})
\]
where $\chi _{Y_+}$ is the characteristic function on $Y_+$. 

\begin{thm}
Let $X, Y _\pm , Z$ and $(G,\phi,c,\tau)$ be as above. Then, we have
\[
\bind [\Hilb, U, \pi , H]= \eind[\hat{\Hilb}, \hat{U} , \hat{\pi} , \hat{H}].
\]
\end{thm}
\begin{proof}
It follows from the definition of $\Psi_{X,G}$ in Corollary \ref{cor:cBC} and the fact that $\mu _{X,G}$ is compatible with the Mayer-Vietoris exact sequences.
\end{proof}

\begin{exmp}\label{exmp:2dIQHE}
Consider the $2$-dimensional IQHE without translation symmetry. As indicated in Table \ref{table:Kitaev}, the bulk and edge indices take values in $\zahl$. Now we relate them with the Hall conductance.

Let $\mathcal{T}_d^\omega$ be as in (\ref{form:cyc}). In the same way as \cite{MR1295473} and \cite{MR1877916}, we obtain the cohomological presentation for the bulk and edge Hall conductance as
\ma{\sigma_{\perp}^b &:= \frac{e^2}{h}2\pi i \mathcal{T}^\omega (P_- [\nabla_1(P_-), \nabla_2(P_-)])=\frac{e^2}{h} \langle [P_-], \mathcal{T} _2^\omega \rangle \\
\sigma _{\perp}^e &:= -\frac{e^2}{h}\lim _{\Delta \to \{ 0 \}}\frac{1}{|\Delta|} \hat{\mathcal{T}}  (\hat{P}_{\Delta} \nabla _1(\hat{H})) =\frac{e^2}{h}\langle [-\exp (\pi i \hat{H})], \hat{\mathcal{T}} _1^\omega \rangle .}
Here note that $P_- \in \Mop _n(\mathcal{B}_u(X))$ since $P_-$ is obtained by a holomorphic functional calculus of the controlled operator $H$. 
By Lemma \ref{lem:pair}, they coincides with $\bind [\Hilb, U,\pi,H]$ and $\eind [\hat{\Hilb}, \hat{U}, \hat{\pi}, \hat{H}]$ respectively. 
By (\ref{form:MVcyc}), we obtain $\sigma _{\perp}^b=\sigma _{\perp}^e$. 
\end{exmp}

\begin{exmp}\label{exmp:oddchiral}
In the same way as above, the odd cyclic cocycle $\mathcal{T}_d^\omega $ gives a homomorphism $\TP (X; {\typeAIII} ) \to \real$, which is explicitly written by Prodan--Schulz-Baldes \cite{mathph14025002} as
\[ \langle [H], \mathcal{T}_d^\omega \rangle := \frac{i(i\pi)^{(d-1)/2}}{d!!}\sum_{\sigma \in \mathfrak{S}_d} (-1)^{|\sigma|} \mathcal{T}^\omega (U^*(\nabla_{\sigma(1)}U) \cdots U^*(\nabla _{\sigma (d)}U))  \]
where $\Phi:=(1+P)/2$ and $U:=\Phi H\gamma \Phi$ is a unitary on $\Phi \Hilb$. By Lemma \ref{lem:pair}, it coincides with the index pairing $\Xi _X^* [H] \in \zahl$. 
It is also shown in \cite{mathph14025002} and related to the index of limit operators. When $d=1$, the bulk-edge correspondence implies that $\langle [H], \mathcal{T}_1^\omega \rangle$ coincides with the Fredholm index of $\hat{H}_0$ where $\hat{H}=\pmx{0 & \hat{H}_0^* \\ \hat{H}_0 & 0}$ with respect to the $\Zt$-grading given by $P$.
\end{exmp}

\begin{exmp}\label{exmp:2dSQHE}
For a $2$-dimensional Quantum Spin-Hall system (quantum systems with type {\typeAII} symmetry) without any other symmetry, the bulk and edge indices take value in $\KAII _2(\real ) =\KQ_2(\real) \cong \Zt$. 
We show that it coincides with the Kane--Mele $\Zt$-invariant or parity of the Spin-Hall conductance.

Assume that a bulk quantum system $(\Hilb, \pi, \varphi , H)$ has another symmetry, that is, $\Hilb =\Hilb _+ \oplus \Hilb _-$ such that $T\Hilb _{\pm}=\Hilb _{\mp}$ and $H\Hilb _\pm=\Hilb _\pm$ (given by spectral subspaces of the spin operator $s_z$).  
In other words, $[H_+] \in \K _0 (C^*_u(X)) \cong \zahl$, where $H_+:=H|_{\Hilb _+}$, satisfies 
\[\lInd _e^{\mathscr{T}} [H_+]=[H] \in \KAII _0 (C^*_u(X)) \cong \KQ_2(\real) \cong \Zt.\]
By Example \ref{exmp:ind}, we obtain
\[
\bind (\Hilb , \pi ,\varphi , H) = 2\pi i \mathcal{T}^\omega (P_- [\nabla_x P_1, \nabla_2 P_-])  \text{ mod $2$},
\]
which coincides with the definition of Kane-Mele invariant in \cite{MR3063955}.  In the same way, we obtain the edge index
\[
\eind (\hat{\Hilb} , \hat{\pi} ,\hat{\varphi }, \hat{H}) = \lim _{\Delta \to \{ 0 \}}\frac{1}{|\Delta|} \hat{\mathcal{T}}  (\hat{P}_{\Delta} \nabla _y\hat{H}) \text{ mod $2$},
\] 
which is the Spin-Hall conductance introduced by Prodan~\cite{MR2525473} modulo $2$. 
We remark that the $\zahl$-valued index is a topological invariant if the additional symmetry is preserved.
In the same way, we obtain $\Zt$-invariants for $1$-dimensional type {\typeDIII} and $3$-dimensional type {\typeCII} topological insulators with additional symmetry by using with odd cocycles.

By Remark \ref{remk:CTvsR} and Proposition \ref{prp:index}, these invariants are the same thing as the index pairing in \cite{GrossmannSchulz-Baldes} and hence the $\Zt$-invariant in \cite{mathph150805485} as is noted in Remark (ii) of \cite{mathph150805485}.
\end{exmp}

\begin{exmp}
Consider $3$-dimensional type {\typeAII} topological insulators. Set $\Pi;=\zahl ^3$ and $X=|\zahl ^3|$. The bulk indices for $\Pi$-invariant and non-invariant systems take values in
\ma{
\KAII^0(B_\Pi)/\mathrm{Triv} &\cong \KQ^0(B_\Pi, \pt) \cong \KQ^2(\pt)^3 \oplus \KQ ^3(\pt) \cong (\zahl _2 )^3 \oplus \zahl _2, \\
\KAII ^{3}(\pt)&\cong \KQ ^3(\pt ) \cong \zahl _2
}
respectively. Moreover, the restriction map 
\[\Res :\KAII_0(C^*(X)^\Pi) \to \KAII_0(C^*(X))\]
is given by the projection onto $\KQ^3(\pt) \cong \Zt$ since the pairing of $\Xi _X$ with elements $\KQ^2(\pt) ^3$ is zero. 
Therefore, we obtain one ``strong'' (i.e.\ robust when the translation symmetry is broken) and three ``weak'' (i.e.\ not preserved when the translation symmetry is broken) $\Zt$-valued topological invariants as is studied by Moore--Balents \cite{MB2007} and Fu--Kane--Mele \cite{FKM2007}. 

Recall that the bulk and edge index maps are isomorphisms by Lemma \ref{lem:uRoe}.
This mean that the edge (and hence the bulk) index is nontrivial if and only if there is a topologically robust ``topological metal'' surface state in the sense of \cite{FKM2007}. 
Moreover, two $3$-dimensional {\typeAII} systems with the same strong invariant are in the same topological phase when we the translation symmetry is broken.
We remark that a model with nontrivial index is found in Section 3.6.3 of \cite{GrossmannSchulz-Baldes}.
\end{exmp}

\begin{exmp}\label{exmp:cryst}
Consider a topological crystalline insulator with disorder which is symmetric under the action of a point group $G$ coming from a spin representation such as rotations. Then, the bulk and edge indices take values in the group as in (\ref{form:group}).  In particular, when we consider the group $C_k:=\zahl /k\zahl$ acting on $\Euc ^2$ by rotations, the indices take values in the group as in (\ref{form:rot}).
\end{exmp}

\begin{exmp}\label{exmp:refl}
Consider a topological crystalline insulator with reflection symmetry, that is, $\mathscr{R}$ acts on $\Euc ^1 \times \Euc ^{d-1}$ as $(x_1,x_2,\dots, x_n) \mapsto (-x_1, x_2, \dots x_n)$. 
Let $\tau$ be an arbitrary $2$-cocycle of $\mathscr{A} \times \mathscr{R}$ such that $(\phi,c,\tau)$ is a twist of $\mathscr{A} \times \mathscr{R}$.
Then, by Corollary \ref{cor:cBC}, the indices take values in $\lu \phi \K ^{\mathscr{A} \times \mathscr{R}}_{d-1,c,\tau}(\Cl _1)$, which is calculated in Example \ref{exmp:reflect}. 
This conclusion is the same thing as the one indicated in table VI of \cite{PhysRevB.88.125129}. We remark that it is exposed in Section IV.B of \cite{PhysRevB.88.125129} that the case of $PR=-RP$ and $TRP=-RPT$ has two components.
\end{exmp}

\begin{remk}\label{remk:conc}
As a concluding remark, we discuss on the meaning of these indices. 
In Example \ref{exmp:2dIQHE} and Example \ref{exmp:2dSQHE}, they are related to the conductance by the TKNN formula. 
Moreover, by Lemma \ref{lem:uRoe}, the edge index is an isomorphism for 
\begin{itemize}
\item arbitrary types when $d=1$ (i.e.\ the dimension of the edge is $0$),
\item arbitrary types other than type {\typeA} when $d=2$,
\item type {\typeAI}, {\typeAII}, {\typeD}, {\typeC}, {\typeBDI} and {\typeCII} when $d=3$.
\end{itemize} 
In other words, the edge index is trivial if and only if the edge Hamiltonian is homotopic to an invertible operator, that is, does not have topologically robust edge channels.
On the other hand, Prodan~\cite{mathph150103479} shows that weak invariants, which is not captured by our index, are protected in $\TP (G,\phi,c,\tau)$ in the case of type A and {\typeAIII} insulators by using lower cyclic cocycles.

Let us consider a variation of the quantum systems using finite $(X,G)$-modules instead of uniformly finite $(X,G)$-modules. 
Then, the equivalence relation of topological phases is also weakened.
Actually, in the same way as Proposition \ref{prp:bulk} and Proposition \ref{prp:edge}, they are isomorphic to a modification of twisted equivariant $\K$-groups of Roe algebras, that is, the range of the bulk and edge indices. 
In other words, two topological phases have the same index if and only if they are equivalent in this rough classification. 
In particular, the lower cocycle of Prodan~\cite{mathph150103479} are not invariants of this rough equivalence relation. 
\end{remk}

\subsection*{Acknowledgement}
The author would like to thank his supervisor Yasuyuki Kawahigashi for his support and encouragement. This work is motivated by his collaborative research with Mikio Furuta, Shin Hayashi, Motoko Kotani, Shinichiro Matsuo and Koji Sato. He would like to thank them for suggesting the problem and many stimulating conversations. He also would like to thank Jean Bellissard, Peter Bouwknegt, Johannes Kellendonk, Yoshiko Ogata, Mathai Varghese, Emil V. Prodan and Guo Chuan Thiang for their helpful conversation. 
He is greatly indebted to referees for their many helpful advices.  
This work was supported by Research Fellow of the JSPS (No.\ 26-7081) and the Program for Leading Graduate Schools, MEXT, Japan.

\bibliographystyle{alpha}
\bibliography{bibABC,bibDEFG,bibHIJK,bibLMN,bibOPQR,bibSTUV,bibWXYZ,arXiv,bulkedge}
\end{document}